\documentclass[11pt,a4paper]{article}

\usepackage{amsfonts,amsmath,latexsym,amsthm}
\usepackage{mathabx}

\usepackage{amssymb}
\usepackage[margin=2.5cm]{geometry} 

\usepackage[UKenglish]{babel}

\usepackage{graphicx}
\usepackage{xspace}
\usepackage{paralist}
\usepackage{enumitem}
\usepackage{subcaption}
\usepackage{balance}
\usepackage[utf8]{inputenc}
\usepackage{todonotes}

\usepackage[nothing]{algorithm}
\usepackage{algorithmicx}
\usepackage[noend]{algpseudocode}
\usepackage{array}

\floatname{algorithm}{Pseudocode}

\newcommand{\para}[1]{\medskip\noindent\textbf{#1}}
\newcommand{\assign}{\,{\gets}\,}
\newcommand{\eat}[1]{}
\newtheorem{claim}{Claim}
\newtheorem{observation}{Observation}

\newcommand{\mycase}[1]{\smallskip\noindent{\bf Case #1:}}
\newcommand{\etal}{{\em et al.}}

\def\advstrat{\textsc{AdvStrategy}}
\def\refineintervals{\textsc{RefineIntervals}}

\def\gap{\operatorname{gap}}
\def\rank{\operatorname{rank}}
\def\next{\operatorname{next}}
\def\prev{\operatorname{prev}}
\def\argmax{\operatorname{arg\,max}}

\newcommand{\polylog}{\operatorname{polylog}}
\newcommand{\poly}{\operatorname{poly}}

\newcommand{\half}{\frac{1}{2}}
\let\eps\varepsilon
\let\rho\varrho
\newcommand{\oneovereps}{\frac{1}{\eps}}
\newcommand{\oneoverdelta}{\frac{1}{\delta}}
\newcommand{\oneoverepssquared}{\frac{1}{\eps^2}}

\newcommand{\ds}{\mathcal{D}}
\newcommand{\Sbar}{\widebar{S}}

\usepackage{hyperref}
\hypersetup{pdftitle={Tight Lower Bound for Comparison-Based Quantile Summaries}}
\hypersetup{pdfauthor={Graham Cormode and Pavel Vesel\'{y}}}

\newcounter{thm}
\numberwithin{thm}{section}
\theoremstyle{plain}
\newtheorem{theorem}[thm]{Theorem}

\newtheorem{lemma}[thm]{Lemma}

\newtheorem{definition}[thm]{Definition}

\usepackage{authblk}

\renewcommand{\O}{\mathcal{O}}

\begin{document}
\allowdisplaybreaks
\title{A Tight Lower Bound for Com\-pa\-ri\-son-Based Quantile Summaries\footnote{A preliminary version of this work has appeared in Proceedings of the 39th ACM SIGMOD-SIGACT-SIGAI Symposium on Principles of Database Systems (PODS ’20).
	Compared to that version, this manuscript contains a corrected proof of the lower bound for relative-error algorithms (Theorem~\ref{thm:biasedQuantilesLB}), using a stronger version of the main result; namely, we prove a lower bound on the \emph{average} space usage of a uniform-error algorithm, not just the maximum space usage.}}

\author[1]{Graham Cormode}
\author[2]{Pavel Vesel{\'{y}}}

\affil[1]{University of Oxford, UK, \texttt{graham.cormode@cs.ox.ac.uk}}
\affil[2]{Charles University, Czech Republic, \texttt{vesely@iuuk.mff.cuni.cz}}

\date{}

\maketitle

\begin{abstract}
Quantiles, such as the median or percentiles, provide concise and useful information
about the distribution of a collection of items, drawn from a totally ordered universe.
We study data structures, called quantile summaries, which keep track
of all quantiles of a stream of items,
up to an error of at most $\eps$.
That is, an $\eps$-approximate quantile summary first
processes a stream and then, given any quantile query $0\le \phi\le 1$, returns an item from the stream,
which is a $\phi'$-quantile for some $\phi' = \phi \pm \eps$.
We focus on comparison-based quantile summaries that can only compare two items
and are otherwise completely oblivious of the universe.

The best such deterministic quantile summary to date,
due to Greenwald and Khanna~\cite{greenwald01_quantile_summaries},
stores at most $\O(\oneovereps\cdot \log \eps N)$ items, where $N$ is the number of items in the stream.
We prove that this space bound is optimal by showing a matching lower bound.
Our result thus rules out the possibility of constructing a
deterministic comparison-based quantile summary in space $f(\eps)\cdot o(\log N)$, for any function
$f$ that does not depend on $N$.
As a corollary, we improve the lower bound for biased quantiles,
which provide a stronger, relative-error guarantee of $(1\pm \eps)\cdot \phi$,
and for other related computational tasks.
\end{abstract}

\section{Introduction}

The streaming model of computation is a useful abstraction to
understand the complexity of working with large volumes of data, too
large to conveniently store.
Efficient algorithms
are known for many basic functions, such as finding frequent items,
computing the number of distinct items, and measuring the empirical
entropy of the data. 
Typically, in the streaming model
we  allow just one pass over the data and a small amount of memory,
i.e., sublinear in the data size.
While computing sums, averages, or counts is trivial
with a constant memory, finding the median, quartiles, percentiles and their generalizations,
quantiles, presents a challenging task.
Indeed, four decades ago, Munro and Paterson~\cite{munro80_exact_quantiles} showed
that finding the exact median in $p$ passes over the data requires $\Omega(N^{1/p})$
memory, where $N$ is the number of items in the stream.
They also provide
a $p$-pass algorithm for selecting the $k$-th smallest item
in space $N^{1/p}\cdot \polylog(N)$, and a $\polylog(N)$-pass 
algorithm running in space $\polylog(N)$.

Thus, either large space, or a large number of passes is
necessary for finding the exact median.
For this reason,
subsequent research has mostly been concerned with the computation of approximate quantiles,
which are often sufficient for applications.
Namely, for a given precision guarantee $\eps > 0$ and a query $\phi\in [0,1]$,
instead of finding the $\phi$-quantile, i.e., the $\lfloor \phi N\rfloor$-th smallest item, we allow the algorithm
to return a $\phi'$-quantile for $\phi' \in [\phi - \eps, \phi + \eps]$;
such an item is called an \emph{$\eps$-approximate $\phi$-quantile}.
In other words, when queried for the $k$-th smallest item (where $k = \lfloor \phi N\rfloor$),
the algorithm may return the $k'$-th smallest item for some $k' \in [k - \eps N, k + \eps N]$.

More precisely, we are interested in a data structure,
called an \emph{$\eps$-approximate quantile summary}, that
processes a stream of items from a totally ordered universe in a single pass.
Then, it returns an $\eps$-approximate $\phi$-quantile
for any query $\phi\in [0,1]$.
We optimize the space used by the quantile summary,
measured in words, where a word can store any item or an integer
with $\O(\log N)$ bits (that is, counters, pointers, etc.).\footnote{
Hence, if instead $b$ bits are needed to  store an item, then the space complexity in bits
	is at most $\max(b, \O(\log N))$ times the space complexity in words.
} We do not assume that items are drawn from a particular distribution,
but rather focus on data independent solutions with worst-case guarantees.
Quantile summaries are a valuable tool, since they immediately provide
solutions for a range of related problems: estimating the cumulative
distribution function; answering rank queries; constructing equi-depth
histograms (where the number of items in each bucket must be
approximately equal); performing Kolmogorov-Smirnov statistical tests
\cite{Lall:15}; and balancing parallel computations~\cite{Tao:Lin:Xiao:13}.

Note that offline, with random access to the whole data set,
we can design an $\eps$-approximate quantile summary with storage cost just
$\left\lceil \frac{1}{2\eps}\right\rceil$. %
We simply select the $\eps$-quantile, the $3\eps$-quantile, the $5\eps$-quantile, 
and so on, and arrange them in a sorted array.
Queries can be answered by returning the $\phi$-quantile of this summary data set. 
Moreover, this is optimal, since there cannot be an interval
$I\subset [0,1]$ of size more than $2\eps$ such that there is no $\phi$-quantile
for any $\phi\in I$ in the quantile summary.

Building on the work of Munro and Paterson~\cite{munro80_exact_quantiles},
Manku, Rajagopalan, and Lindsay~\cite{manku98_QS}
designed a (streaming) quantile summary which uses space $\O(\oneovereps\cdot \log^2 \eps N)$,
although it relies on the advance knowledge of the stream length $N$.
Then, shaving off one log factor, Greenwald and Khanna~\cite{greenwald01_quantile_summaries}
gave an $\eps$-approximate quantile summary,
which needs just $\O(\oneovereps\cdot \log \eps N)$ words and does not require any advance information about the stream.
Both of these deterministic algorithms work for any universe with a total ordering as they
just need to do comparisons of the items. We call such an algorithm \emph{comparison-based}.

The question of whether one can design a 1-pass deterministic algorithm that runs
in a constant space for a constant $\eps$ has been open for a long time,
as highlighted by the first author
in 2006~\cite{cormode06_open_problem2}.
Following the above discussion, there is a trivial lower bound of
$\Omega(\oneovereps)$ that holds even offline.
This was the best known lower bound
until 2010 when Hung and Ting~\cite{hung10_qs_lower_bound} 
proved that a deterministic comparison-based algorithm needs space $\Omega(\oneovereps\cdot\log\oneovereps)$.

We significantly improve upon that result by showing that any deterministic
comparison-based data structure providing $\eps$-app\-ro\-xi\-ma\-te quantiles
needs to use $\Omega(\oneovereps\cdot \log \eps N)$ memory on the worst-case
input stream.
Our lower bound thus matches the Greenwald and Khanna's result, up to a constant factor,
and in particular, it rules out an algorithm running in space $f(\eps)\cdot o(\log N)$,
for any function $f$ that does not depend on $N$.
It also follows that a comparison-based data structure with $o(\oneovereps\cdot \log \eps N)$ memory 
must fail to provide a $\phi$-quantile for some $\phi\in [0, 1]$.
Using a standard reduction (appending more items to the end of the stream), 
this implies that there is no deterministic comparison-based streaming
algorithm that returns an $\eps$-approximate median and
uses $o(\oneovereps\cdot \log \eps N)$ memory.
Applying a reduction from~\cite{karnin16_optimal_rand_quantile_summaries}, this yields a lower bound of
$\Omega(\oneovereps\cdot \log \log \oneoverdelta)$
for any randomized comparison-based algorithm. 
We refer to Section~\ref{sec:conclusions}
for a discussion of this and other  corollaries of our result,
including a new lower bound for algorithms achieving a stronger relative-error guarantee.

\subsection{Overview and Comparison to Prior Bounds}

Let $\ds$ be a deterministic comparison-based quantile summary.
From a high-level point of view, we prove the space lower bound for $\ds$ by
constructing two streams $\pi$ and $\rho$ satisfying two opposing constraints:
On one hand, the behavior of
$\ds$ on these streams is the same, implying that the memory states
after processing $\pi$ and $\rho$ are the same, up to an order-preserving
renaming of the stored items. For this reason, $\pi$ and $\rho$ are
called \emph{indistinguishable}. On the other hand, the adversary introduces
as much uncertainty as possible. Namely, it makes the difference between the rank
of a stored item with respect to (w.r.t.) $\pi$ and the rank of the next stored item
w.r.t.\ $\rho$ as large as possible, where the rank of an item w.r.t.\ stream $\sigma$
is its position in the ordering of $\sigma$.
If this difference, which we call the ``gap'', is too large, then
$\ds$ fails to provide an $\eps$-approximate $\phi$-quantile for some $\phi\in [0,1]$.
The crucial part of our lower bound proof is to construct
the two streams in a way that yields a good trade-off between
the number of items stored by the algorithm and the largest
gap introduced.

While the previous lower bound of $\Omega(\oneovereps\cdot \log \oneovereps)$~\cite{hung10_qs_lower_bound}
is in the same computational model, and also works by creating
indistinguishable streams with as much uncertainty as possible,
our approach is substantially different.
Mainly, the construction by Hung and Ting~\cite{hung10_qs_lower_bound} is inherently sequential
as it works in $m \approx \oneovereps \log \oneovereps$ iterations
and appends $\O(m)$ items in each iteration to the streams constructed
(and moreover, up to $\O(m)$ new streams are created from each former stream in each iteration).
Thus, their construction produces (a large number of) indistinguishable streams of length
$\Theta\left(\left(\oneovereps \log \oneovereps\right)^2\right)$.
Furthermore, having the number of iterations equal to the number of items 
appended during each iteration (up to a constant factor) is crucial 
for the analysis in~\cite{hung10_qs_lower_bound}. 

In contrast, our construction is naturally specified in a
recursive way, and it 
produces just two indistinguishable streams of length $N$ for any $N = \Omega(\oneovereps)$.
For $N \approx \left(\oneovereps\right)^2$,
our lower bound of $\Omega(\oneovereps\cdot  \log \eps N)$ implies the previous 
one of $\Omega(\oneovereps\cdot  \log \oneovereps)$,
and hence for higher $N$, our lower bound is strictly stronger than the previous one.

The value in using a recursive construction is as follows:
The construction produces two indistinguishable streams of length $\oneovereps\cdot 2^k$ for an integer $k\ge 1$,
and we need to prove that the quantile summary $\ds$ must store at least $c\cdot \oneovereps\cdot k$
items while processing one of these streams, for a constant $c>0$.
The first half of the streams is constructed recursively,
so $\ds$ needs to store at least $c\cdot \oneovereps\cdot (k-1)$ items while processing the first half
of either of these two streams (using an induction on $k$).
If it already stores at least $c\cdot \oneovereps\cdot k$ items on the first half,
then we are done.
Otherwise, our inductive argument yields that there must a substantial uncertainty
introduced while processing the first half,
which we use in the recursive construction of the second half of the streams.
Then our aim will be to show that, while processing the second 
half, $\ds$ needs to store $c\cdot \oneovereps\cdot (k-1)$ items from the second half,
by induction, and $c\cdot \oneovereps$ items from the first half, by a simple bound.
Hence, it stores $c\cdot \oneovereps\cdot k$ items overall.
However, using the inductive argument on the second half brings some technical difficulties,
since the streams already contain items from the first half.
Our analysis shows a space lower bound, called the ``space-gap inequality'', that depends on the uncertainty introduced
on a particular part of the stream, and this inequality is amenable to a proof by induction.

\para{Final versus working space usage.}
We note that we prove a lower bound on the space usage of any deterministic comparison-based quantile summary while processing the input stream.
However, after processing the stream, the summary may be compressed to $O(1/\eps)$ space,
simply by processing the stream with smaller accuracy parameter $\eps' = \eps/2$
and then selecting $1/\eps'$ many $\eps'$-approximate quantiles, which form a final summary
that achieves uniform $\pm \eps\cdot N$ error.

Nevertheless, we show that on average over all steps,
any deterministic comparison-based quantile summary must use $\Omega(\oneovereps\cdot \log \eps N)$ memory words, which will be crucially needed for the relative-error lower bound.

\para{Organization of the paper.}
In Section~\ref{sec:compModel}, we start by describing the formal computational model
in which our lower bound holds
and formally stating our result. In Section~\ref{sec:indistinguishableStreams},
we introduce indistinguishable streams, and in Section~\ref{sec:construction}
we describe our construction. Then, in Section~\ref{sec:spaceGapInquality}
we inductively prove the crucial inequality between the space and the largest gap (the uncertainty),
which implies the lower bound. Finally, in Section~\ref{sec:conclusions}
we give corollaries of the construction and discuss related open problems.

\subsection{Related Work}

The Greenwald-Khanna (GK) algorithm~\cite{greenwald01_quantile_summaries}
is generally regarded as the best deterministic quantile summary.
The space bound of $O(\oneovereps \cdot \log \eps N)$ follows from a
somewhat involved proof, and it has been questioned whether this
approach could be simplified or improved.
Our work answers this second question in the negative,
while Gribelyuk \etal~\cite{GribelyukSWY24} recently provided a simpler proof of the optimal bound for a variant of GK.
For a known universe $U$ of bounded size, 
Shrivastava \etal~\cite{shrivastava04_bounded_universe_qs} designed a quantile summary q-digest
using $\O(\oneovereps\cdot \log |U|)$ words.
Note that their algorithm is 
not comparison-based and so the result is incomparable to the upper bound of $\O(\oneovereps\cdot\log \eps N)$.
We are not aware of any lower bound which holds for a known universe of bounded size,
apart from the trivial bound $\Omega(\oneovereps)$.
Recently, Gupta, Singhal, and Wu~\cite{GuptaSW24} designed a compressed variant of q-digest that uses only $\O(\oneovereps)$ memory words.

If we tolerate randomization, %
it is possible to design quantile summaries with space close to $\oneovereps$.
After a sequence of
improvements~\cite{manku99_sampling_QS,agarwal13_mergeable_summaries,luo16_quantiles_experimental,felber15_randomized_QS},
Karnin, Lang, and Liberty~\cite{karnin16_optimal_rand_quantile_summaries}
designed a randomized comparison-based quantile summary with space bounded by
$\O(\oneovereps\cdot \log\log \frac{1}{\eps \delta})$,
where $\delta$ is the probability of not returning an $\eps$-app\-ro\-xi\-ma\-te $\phi$-quantile
for some $\phi$.
They also provide a reduction to transform the deterministic
$\Omega(\oneovereps\cdot \log\oneovereps)$ lower bound into a randomized lower bound of $\Omega(\oneovereps\cdot \log\log \frac{1}{\delta})$ for $\delta < 1/N!$,
implying optimality of their approach in the comparison-based model
for an exponentially small $\delta$.
We discuss further how the deterministic and randomized
lower bounds relate in Section~\ref{sec:conclusions}.

Luo \etal~\cite{luo16_quantiles_experimental} compared quantile summaries
experimentally and also provided a simple randomized algorithm with a good practical performance.
This paper studies not only streaming algorithms for insertion-only streams (i.e., the cash register model),
but also for  turnstile streams, in which items may depart.
Note that any algorithm for turnstile streams inherently relies on the bounded
size of the universe.
We refer the interested reader to the survey of Greenwald and Khanna~\cite{greenwald16_quantiles_survey}
for a description of both deterministic and randomized algorithms, together
with algorithms for turnstile streams, the sliding window model, and distributed algorithms.

Other results arise when relaxing the requirement for correctness
under adversarial order to assuming that the input arrives in a random
order. 
For random-order streams, Guha and McGregor~\cite{guha09_randomoder_quantiles} studied algorithms for exact
and approximate selection of quantiles.
Among other things, they gave an algorithm
for finding the exact $\phi$-quantile in space $\polylog(N)$ using $\O(\log \log N)$
passes over a random-order stream, while with $\polylog(N)$ memory 
we need to do $\Omega(\log N / \log \log N)$ passes on the worst-case stream.
For the one-pass random-order setting, they design a logarithmic-space algorithm which, for a given $k$,
returns an item of rank $k\pm O(\sqrt{k}\cdot \log n\cdot \log \delta^{-1})$ with probability at least $1-\delta$.
The Shifting Sands algorithm~\cite{McGregor:Valiant:12} reduces the
magnitude of the error for the median from $O(n^{1/2})$ to $O(n^{1/3})$. 
Since our lower bound relies on carefully constructing an adversarial
input sequence, it does not apply to this random order model.

\section{Computational Model}\label{sec:compModel}

We present our lower bounds in a comparison-based model of
computation, in line with prior work, most notably that of 
Hung and Ting~\cite{hung10_qs_lower_bound}.
We assume that the items forming the input stream are drawn from a
totally ordered universe $U$, 
about which the algorithm has no further information.
The only allowed operations on items are to perform an equality test
or a comparison of two given items.
This restriction specifically rules out manipulations which try to combine
multiple items into a single storage location, or replace a group of
items with an ``average'' representative.  
We assume that the universe is unbounded and continuous in the sense
that any non-empty open interval contains an unbounded number of
items.
This fact is relied on in our proof to be able to draw new elements
falling between any previously observed pair. 
An example of such a universe is a large enough set of long incompressible strings, ordered lexicographically
(where the continuous assumption may be achieved by making the strings even longer).

Let $\ds$ be a deterministic data structure for processing a stream of items,
i.e., a sequence of items arriving one by one.
We make the following assumptions about the memory contents of $\ds$.
The memory used by $\ds$ will contain some items from the stream,
each considered to occupy one memory cell, and some other information
which could include lower and upper bounds on the ranks of stored items, counters, etc.
However, we assume that the memory does not contain the result
of any operation applied on any $k\ge 1$ items from the stream,
apart from a comparison and the equality test (as other operations are prohibited by our model).
Thus, we can partition the \emph{memory state} into a pair $M = (I, G)$,
where $I$ is the \emph{item array} for storing items from the
input, indexed from $1$,
and there are no items stored in the \emph{general memory} $G$.

We give our lower bound on the memory size only in terms of $|I|$, the
number of items stored, and ignore the size of $G$.
For simplicity, we assume without loss of generality that the contents of $I$ are sorted non-decreasingly, i.e.,
$I[1]\le I[2]\le \cdots$.
If this were not case, we could equivalently apply an in-place sorting
algorithm after processing each item, while the information
potentially encoded in the former ordering of $I$ can be retained in
$G$ whose size we do not measure.
Finally, we can assume that the minimum and maximum elements of the
input stream are always maintained, with at most a constant additional
storage space.

Summarizing, we have the following definition.

\begin{definition}\label{def:comparisonBased}
We say that a quantile summary $\ds$ is \emph{comparison-based} 
if the following holds:
\begin{enumerate}[label=(\roman*),nosep]
\item $\ds$ does not perform any operation on items from the stream,
apart from a comparison and the equality test.
\item The memory of $\ds$ is divided into the \emph{item array} $I$, which stores only items
that have already occurred in the stream (sorted non-decreasingly),
and \emph{general memory} $G$, which does not contain any item identifier.
Furthermore, once an item is removed from $I$, it cannot be added back to $I$, unless
it appears in the stream again.
\item Given the $i$-th item $a_i$ from the input stream, \label{itm:compArrival}
the computation of $\ds$ is determined solely by the results of
comparisons between $a_i$ and $I[j]$, for $j= 1,\dots,|I|$,
the number $|I|$ of items stored, and the contents of the general memory $G$.
\item Given a quantile query $0\le \phi\le 1$, the computation of $\ds$ is determined solely by 
the number of items stored ($|I|$), and the contents of the general memory $G$.
Moreover, $\ds$ can only return one of the items stored in $I$.
\end{enumerate}
\end{definition}

Note that quantile summaries satisfying Definition~\ref{def:comparisonBased}
include the Greenwald-Khanna algorithm~\cite{greenwald01_quantile_summaries}
as well as many other deterministic~\cite{greenwald16_quantiles_survey,manku98_QS,munro80_exact_quantiles} and randomized quantile summaries~\cite{manku99_sampling_QS,agarwal13_mergeable_summaries,
luo16_quantiles_experimental,felber15_randomized_QS,karnin16_optimal_rand_quantile_summaries}.
On the other hand, the q-digest structure~\cite{shrivastava04_bounded_universe_qs}
is not comparison-based, since it relies on building a binary tree over $U$
and since it can actually return an item that did not occur in the stream,
neither of which is allowed by Definition~\ref{def:comparisonBased}.
Thus, our lower bound does not apply to this algorithm and indeed, for $N\gg |U|$,
its space requirement of $\O(\oneovereps\cdot \log |U|)$ words may be substantially smaller
than $\Omega(\oneovereps\cdot \log \eps N)$.

We are now ready to state our main result formally.

\begin{theorem}\label{thm:main}
For any $0 < \eps < \frac1{32}$,
there is no deterministic comparison-based $\eps$-approximate quantile summary
which stores $o(\oneovereps\cdot \log \eps N)$ items on any input stream of length $N$.
Furthermore, the same lower bound holds for the average space usage, where the average is taken over all stream updates.
\end{theorem}

Fix the approximation guarantee $0 < \eps < \frac1{32}$ and assume for simplicity that 
$\oneovereps$ is an integer.
Let $\ds$ be a fixed deterministic comparison-based $\eps$-approximate quantile summary.
We show that for any integer $k\ge 1$, data structure $\ds$ needs to store 
at least $\Omega(\oneovereps\cdot k)$ items from some input stream of length $N_k := \oneovereps\cdot 2^k$
(thus, we have $\log_2 \eps N_k = k$).

\para{Notation and conventions.}
We assume that $\ds$ starts with an empty memory state $M_\emptyset = (I_\emptyset, G_\emptyset)$
with $|I_\emptyset| = 0$.
For an item $a$, let $\ds(M, a)$ be the resulting memory state
after processing item $a$ if the memory state was $M$ before processing $a$. Moreover,
for a stream $\sigma = a_1, \dots, a_N$, let
$\ds(M, \sigma) = \ds(\dots \ds(\ds(M, a_1), a_2), \dots, a_N)$ be the memory state after processing stream $\sigma$.
For brevity, we use $(I_\sigma, G_\sigma) = \ds(M_\emptyset, \sigma)$,
or just $I_\sigma$ for the item array after processing stream $\sigma$.

When referring to the order of a set of items, we always mean the non-decreasing order.
For an item $a$ in stream $\sigma$, let $\rank_\sigma(a)$ be the rank of $a$ in the order of $\sigma$,
i.e., the position of $a$ in the ordering of $\sigma$. 
In our construction, all items in each of the streams will be
distinct, thus $\rank_\sigma(a)$ is well-defined and equal to one more than the number of items
that are strictly smaller than $a$.

\section{Indistinguishable Streams}\label{sec:indistinguishableStreams}
We start by defining an equivalence of memory states of the fixed summary $\ds$,
which captures their equality up to renaming stored items.
Then, we give the definition of indistinguishable streams.

\begin{definition}\label{def:equivalentStates}
Two memory states $(I_1,G_1)$ and $(I_2,G_2)$ are said to be \emph{equivalent} if (i) $|I_1| = |I_2|$,
i.e., the number of items stored is the same, and
(ii) $G_1 = G_2$.
\end{definition}

\begin{definition}\label{def:indistinguishableStreams}
We say that two streams $\pi = a_1 a_2 \dots a_N$ and $\rho = b_1 b_2 \dots b_N$ of length $N$
are \emph{indistinguishable} for $\ds$
if (1) the final memory states $(I_\pi, G_\pi)$ and $(I_\rho, G_\rho)$
are equivalent, and (2) for any $1\le i\le |I_\pi| = |I_\rho|$,
there exists $1\le j\le N$ such that both $I_\pi[i] = a_j$ and $I_\rho[i] = b_j$.
\end{definition}

We remark that condition~(2) is implied by~(1)
if the positions of stored items in the stream are retained in the general memory,
but we make this property explicit as we shall use it later.
In the following, let $\pi$ and $\rho$ be two indistinguishable streams with $N$ items.
Note that, after $\ds$ processes one of $\pi$ and $\rho$
and receives a quantile query $0\le \phi\le 1$,
$\ds$ must return the $i$-th item of array $I$ for some $i$, regardless of whether the stream was $\pi$ or $\rho$.
This follows, since $\ds$ can make its decisions based on the values
in $G$, which are identical in both cases, and operations on 
values in $I$, which are indistinguishable under the comparison-based
model. 

For any $k\ge 1$, our general approach is to  recursively construct two
streams $\pi_k$ and $\rho_k$ of length $N_k$ that satisfy two
constraints set in opposition to each other:
They are indistinguishable for $\ds$, but at the same time, for some $j$, the rank of $I_\pi[j]$ in stream $\pi$
and the rank of $I_\rho[j+1]$ in stream $\rho$ 
are as different as possible --- we call this difference the ``gap''.
The latter constraint is captured by the following definition.

\begin{definition}
We define the \emph{largest gap} between indistinguishable streams
$\pi$ and $\rho$ (for $\ds$) as 
\begin{align*}\gap(\pi, \rho) = \max_{1\le i < |I_\pi|}\,
	 \max\big(\, & \rank_\pi(I_\pi[i+1]) - \rank_\rho(I_\rho[i]),\, \\
         & \rank_\rho(I_\rho[i+1]) - \rank_\pi(I_\pi[i]) \,\big)\,.
         \end{align*}
\end{definition}

As we assume that $I$ is sorted, $I_\pi[i+1]$ is the next
stored item after $I_\pi[i]$ in the ordering of $I_\pi$.
In the construction in Section~\ref{sec:construction},
we will also ensure that $\rank_\pi(I_\pi[i]) \le \rank_\rho(I_\rho[i])$ for any $1\le i\le |I_\pi|$.
Hence, we can simplify to \[\gap(\pi, \rho) = \max_i \rank_\rho(I_\rho[i+1]) -
\rank_\pi(I_\pi[i]).\]
We also have that 
$\gap(\pi, \rho)\ge \gap(\pi, \pi)$, which follows, since for any $i$
it holds by construction that \[\rank_\rho(I_\rho[i+1]) - \rank_\pi(I_\pi[i]) \ge
\rank_\pi(I_\pi[i+1]) - \rank_\pi(I_\pi[i]).\]

\begin{lemma}\label{lem:indistStreamsGapUB}
If $\ds$ is an $\eps$-approximate quantile summary, then $\gap(\pi, \rho) \le 2\eps N$.
\end{lemma}

\begin{proof}
Suppose that $\gap(\pi, \rho) > 2\eps N$.
We show that $\ds$ fails to provide an $\eps$-approximate $\phi$-quantile for some $0\le \phi\le 1$,
which is a contradiction.
Namely, because $\gap(\pi, \rho) > 2\eps N$, there is $1\le i < |I_\pi| = |I_\rho|$
such that $\rank_\rho(I_\rho[i+1]) - \rank_\pi(I_\pi[i]) > 2\eps N$.
Let $\phi$ be such that
\[\phi\cdot N = \half \big(\rank_\rho(I_\rho[i+1]) + \rank_\pi(I_\pi[i])\big),\] i.e.,
$\phi\cdot N$ is in the middle of the ``gap''. 
Since streams $\pi$ and $\rho$ are indistinguishable and $\ds$ is comparison-based,
given query $\phi$, $\ds$ must return the $j$-th item of item array $I$ for some $j$,
regardless of whether the stream is $\pi$ or $\rho$.
Observe that if $j\le i$ and the input stream is $\pi$,
item $I_\pi[j]$ does not meet the requirements to be an $\eps$-approximate $\phi$-quantile of items in $\pi$.
Otherwise, when $j > i$, then item $I_\rho[j]$ is not an $\eps$-approximate $\phi$-quantile of stream $\rho$.
In either case, we get a contradiction.
\end{proof}

As the minimum and maximum elements of stream $\pi$ are in $I_\pi$, it holds that $\gap(\pi, \pi) \ge N / |I_\pi|$,
thus the number of stored items is at least $N / \gap(\pi, \pi) \ge N / \gap(\pi, \rho) \ge \frac{1}{2\eps}$,
where the last inequality is by Lemma~\ref{lem:indistStreamsGapUB}.
This gives an initial lower bound of $\Omega(\oneovereps)$ space.
Our construction of adversarial inputs for $\ds$ in the next section
increases this bound. 

\section{Recursive Construction of Indistinguishable Streams}\label{sec:construction}

\subsection{Intuition}
We will define our construction of the two streams $\pi$ and $\rho$ using a
recursive adversarial procedure for generating items into the two streams.
This procedure tries to make the gap as large as possible,
but ensures that they are indistinguishable.
It helps to consider  the recursion tree.
This tree is a full binary tree with $k$ levels, with the root at level $1$
and thus with $2^{k-1}$ leaves at level $k$. 
In each leaf, $2/\eps$ items are appended to the stream, while
the adversary generates no items in internal (i.e., non-leaf) nodes.
The construction performs the in-order traversal of the recursion tree.

One of the key concepts needed is the maintenance of two open
intervals during the construction, one for stream $\pi$, denoted $(\ell_\pi, r_\pi)$, and the other for stream $\rho$, denoted $(\ell_\rho, r_\rho)$.
Initially, these intervals cover the whole universe,
but they are refined in each internal node of the recursion tree.
More precisely, consider the execution in an internal node $v$ at level $i$ of the recursion tree.
We first execute the left subtree, which generates $\oneovereps\cdot 2^{i-1}$
items into the streams inside the current intervals.
We then identify the largest gap inside the current intervals w.r.t.\ item arrays of $\ds$
after processing streams $\pi$ and $\rho$ (more precisely, after $\ds$
has completed processing the prefixes of $\pi$ and $\rho$
constructed so far).
Having the largest gap, we identify new open intervals
for $\pi$ and $\rho$ in ``extreme regions'' of this gap, so that they do not
contain any item so far.
We explain this subroutine in greater detail when describing procedure \refineintervals{}.
We choose these intervals so that indistinguishability of the streams is preserved,
while the rank difference between the two streams (the uncertainty) is maximized.
The execution of the procedure in node $v$ ends by executing the right subtree of $v$,
which generates a further  $\oneovereps\cdot 2^{i-1}$
items into the new, refined intervals of the two streams.
Recall that we consider the universe of items to be continuous, namely, 
that we can always generate sufficiently many items within both of the new intervals.

\subsection{Notation}
For an item $a$ in stream $\sigma$, let $\next(\sigma, a)$ be the next item 
in the ordering of $\sigma$, i.e., the smallest item in $\sigma$ that is larger than $a$
(we never invoke $\next(\sigma, a)$ when $a$ is the largest item in $\sigma$).
Similarly, for an item $b$ in stream $\sigma$, let $\prev(\sigma, b)$ be the previous item in 
the ordering of $\sigma$ (left undefined for the smallest item in $\sigma$).
Note that $\next(\sigma, a)$ or $\prev(\sigma, b)$ may well \emph{not} be stored by $\ds$.

For an interval $(\ell, r)$ of items and an array $I$ of items, we use 
$I^{(\ell, r)}$ to denote the restriction of $I$ to $(\ell, r)$, enclosed by $\ell$ and $r$.
That is, $I^{(\ell, r)}$ is the array of items $\ell, I[i], I[i+1], \dots, I[j], r$,
where $i$ and $j$ are the minimal and maximal indexes of an item in $I$
that falls within the interval $(\ell, r)$, respectively.
Items in $I^{(\ell, r)}$ are taken to be sorted and indexed from 1.
Recall also that by our convention, $I_\sigma$ is the item array after processing some stream $\sigma$.

\begin{algorithm}[t]
\caption{Adversarial procedure \refineintervals{}}
\label{alg:refineintervals}
\begin{algorithmic}[1]
\Require{Streams $\pi$ and $\rho$ and intervals $(\ell_\pi, r_\pi)$ and $(\ell_\rho, r_\rho)$ of items
such that:
\begin{enumerate}[nosep,label=(\roman*)]
\item $\pi$ and $\rho$ are indistinguishable, and
\item only the last $N'\ge 2$ items from $\pi$ and $\rho$ are from intervals $(\ell_\pi, r_\pi)$ and $(\ell_\rho, r_\rho)$, respectively
\end{enumerate}
}
\Ensure{Intervals $(\alpha_\pi, \beta_\pi)\subset (\ell_\pi, r_\pi)$ and
	$(\alpha_\rho, \beta_\rho)\subset (\ell_\rho, r_\rho)$}
	\State{$I'_{\pi} \assign I^{(\ell_\pi, r_\pi)}_{\pi}$ and $I'_{\rho} \assign I^{(\ell_\rho, r_\rho)}_{\rho}$}
		\label{algLn:restrItemArrays}
	\State{$\displaystyle{i\assign \argmax_{1\le i < |I'_{\rho}|} \rank_{\rho}(I'_{\rho}[i+1]) - \rank_{\pi}(I'_{\pi}[i])}$}
\newline
	\Comment{Position of the largest gap in intervals $(\ell_\pi, r_\pi)$ and $(\ell_\rho, r_\rho)$}
		\label{algLn:gapDef}
	\State{$(\alpha_\pi, \beta_\pi) \assign \big(I'_{\pi}[i], \next(\pi, I'_{\pi}[i])\big)$}
		\Comment{New interval for $\pi$}
		\label{algLn:intervalForPi}
	\State{$(\alpha_\rho, \beta_\rho) \assign \big(\prev(\rho, I'_{\rho}[i+1]), I'_{\rho}[i+1]\big)$}
		\Comment{New interval for $\rho$}
		\label{algLn:intervalForRho}
	\State{\Return $(\alpha_\pi, \beta_\pi)$ and $(\alpha_\rho, \beta_\rho)$}
\end{algorithmic}
\end{algorithm}

\begin{figure}[t]
\centerline{\includegraphics[scale=0.8]{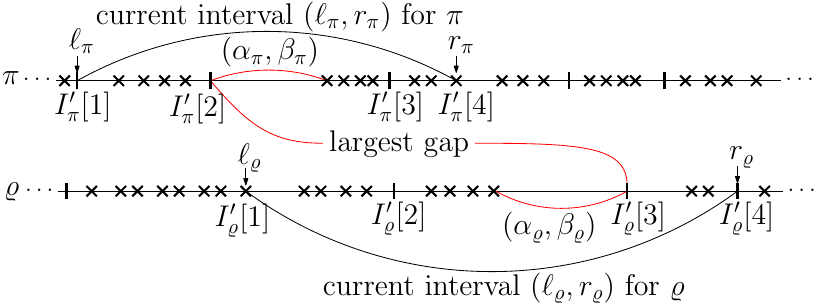}}
\caption{An illustration of the largest gap computation.}
\label{fig:gap}
\end{figure}

\subsection{Procedure \refineintervals{}}
We next describe our procedure to find the largest gap and refine the intervals,
defined in Pseudocode~\ref{alg:refineintervals}.
It takes as input indistinguishable streams $\pi$ and $\sigma$
and two open intervals $(\ell_\pi, r_\pi)$ and $(\ell_\rho, r_\rho)$ of the universe,
such that intervals $(\ell_\pi, r_\pi)$ and $(\ell_\rho, r_\rho)$ contain 
only the last $N'$ items from $\pi$ and $\rho$, respectively, for some $N'\ge 2$.
Note that $I'_{\pi}$ and $I'_{\rho}$ are the item arrays of $\ds$ for $\pi$ and $\rho$
restricted to the intervals $(\ell_\pi, r_\pi)$ and $(\ell_\rho, r_\rho)$, respectively, as defined above.
In these restricted arrays we find the largest gap (in line~\ref{algLn:gapDef}),
which is determined by the largest rank difference of consecutive items in the two arrays.
Finally, in lines~\ref{algLn:intervalForPi} and~\ref{algLn:intervalForRho}, we define
new, refined intervals in the extreme regions of the gap.
To be  precise, the new open interval for $\pi$ is between item $I'_{\pi}[i]$ (whose rank is used
to determine the largest gap) and the next item after $I'_{\pi}[i]$ in the ordering of $\pi$, i.e., $\next(\pi, I'_{\pi}[i])$.
The new open interval for $\rho$ is defined in a similar way:
It is between item $I'_{\rho}[i+1]$ (used to determine the largest gap)
and the item that precedes it in the ordering of $\rho$, i.e., $\prev(\rho, I'_{\rho}[i+1])$.

\medskip

In Figure~\ref{fig:gap} we give an illustration.
In this figure, the items in the streams are real numbers and we
depict them on the real line, the top one for $\pi$ and the bottom one for $\rho$.
Each item is represented either by a short line segment if it is stored in the item array,
or by a cross otherwise (indicating that it has been ``forgotten'' by $\ds$).
The procedure looks for the largest gap only within the current intervals
$(\ell_\pi, r_\pi)$ and $(\ell_\rho, r_\rho)$.
The ranks of items in the restricted item arrays (i.e.,
disregarding items outside the current intervals) can be verified to be
$1, 6, 11,$ and $14$ w.r.t.\ both streams.
(Note that $r_\pi$ is the last item in the restricted item array $I'_\pi$, even though
it was discarded from the whole item array $I_\pi$ by the algorithm, and similarly
for $\ell_\rho$ and $I'_\rho$.)
Thus the largest gap size is $5$ items,
and is found between the second item in the restricted item array $I'_\pi$ and the third item in $I'_\rho$,
as highlighted in the figure.
In this example, there is another, equal sized, gap between the first and second item in these arrays;
ties can be broken arbitrarily. 
The new intervals in the extreme regions of the largest gap
are depicted as well.
\qed

\medskip

\footnotetext{The minimum over an empty set is defined
	arbitrarily to be $\infty$.} 

\begin{algorithm*}[t]
\caption{Adversarial procedure \advstrat{}}
\label{alg:advstrat}
\begin{algorithmic}[1]
\Require{Integer $k\ge 1$, streams $\pi$ and $\rho$, 
	and intervals $(\ell_\pi, r_\pi)$ and $(\ell_\rho, r_\rho)$ of items such that:
	\begin{enumerate}[nosep,label=(\roman*)]
	\item $\pi$ and $\rho$ are indistinguishable, 
	\item $\pi$ contains no item from $(\ell_\pi, r_\pi)$ and $\rho$ contains no item from $(\ell_\rho, r_\rho)$, and
	\item for any $a\in (\ell_\pi, r_\pi)$ and $b\in (\ell_\rho, r_\rho)$, it holds that
		$\min \{ i | a \le I_{\pi}[i] \} = \min \{ i | b \le I_{\rho}[i] \}$\footnotemark
	\end{enumerate}
}
\Ensure{Streams $\pi'' = \pi \pi_k$ and $\rho'' = \rho \rho_k$, where $\pi_k$ and $\rho_k$ are substreams with $\oneovereps\cdot 2^k$ items
from $(\ell_\pi, r_\pi)$ and $(\ell_\rho, r_\rho)$, respectively}
\If{k = 1} \Comment{Leaf node of the recursion tree}
	\State{$\pi'' \assign$ stream $\pi$ followed by $2 / \eps$ items from interval $(\ell_\pi, r_\pi)$, in order }
	\State{$\rho'' \assign$ stream $\rho$ followed by $2 / \eps$ items from interval $(\ell_\rho, r_\rho)$, in order}
	\State{\Return{} streams $\pi''$ and $\rho''$}
\Else	\Comment{Internal node of the recursion tree}
	\State{$(\pi', \rho') \assign \advstrat{}\big(k - 1, \pi, \rho, (\ell_\pi, r_\pi), (\ell_\rho, r_\rho)\big)$}
		\label{algLn:stage1}
	\State{$(\alpha_\pi, \beta_\pi), (\alpha_\rho, \beta_\rho) \assign \refineintervals(\pi', \rho', (\ell_\pi, r_\pi), (\ell_\rho, r_\rho))$}
		\label{algLn:refineIntervals}
	\State{\Return{} $(\pi'', \rho'') \assign \advstrat{}\big(k - 1, \pi', \rho', (\alpha_\pi, \beta_\pi), (\alpha_\rho, \beta_\rho))\big)$}
		\label{algLn:stage2}
\EndIf
\end{algorithmic}
\end{algorithm*}

We claim that in the \refineintervals{} procedure $|I'_{\pi}| = |I'_{\rho}|$, which implies that the largest gap in line~\ref{algLn:gapDef}
is well-defined. Let $\pi = a_1 \dots a_{N}$ and $\rho = b_1 \dots b_{N}$ be the items in
streams $\pi$ and $\rho$, respectively. Since streams $\pi$ and $\rho$ are indistinguishable,
condition~(2) in Definition~\ref{def:indistinguishableStreams} implies that
for any $1\le i\le |I_{\pi}| = |I_{\rho}|$ (where $I_{\pi}$ and $I_{\rho}$ are the full item arrays),
there exists $1\le j\le N$ such that both $I_{\pi}[i] = a_j$ and $I_{\rho}[i] = b_j$.
As only the last $N'$ items of $\pi$ and of $\sigma$ are from intervals $(\ell_\pi, r_\pi)$ and $(\ell_\rho, r_\rho)$, respectively,
we %
 obtain that the restricted item arrays $I'_{\pi} = I^{(\ell_\pi, r_\pi)}_{\pi}$ and
$I'_{\rho} = I^{(\ell_\rho, r_\rho)}_{\rho}$ must have the same size,
proving the claim.

Finally, we show two properties that will be useful later
and follow directly from the definition of the new intervals.

\begin{observation}\label{obs:refineIntervalsOutput}
For intervals $(\alpha_\pi, \beta_\pi)$ and $(\alpha_\rho, \beta_\rho)$
returned by $\refineintervals\,(\pi, \rho, (\ell_\pi, r_\pi), (\ell_\rho, r_\rho))$,
it holds that 
\begin{enumerate}[nosep,label=(\roman*)]
\item $\pi$ contains no item in the interval $(\alpha_\pi, \beta_\pi)$ and
$\rho$ contains no item in the interval $(\alpha_\rho, \beta_\rho)$; and
\item for any $a\in (\alpha_\pi, \beta_\pi)$ and $b\in (\alpha_\rho, \beta_\rho)$
we have that $\min \{ i | a \le I_{\pi}[i] \} = \min \{ i | b \le
I_{\rho}[i] \}$.
\end{enumerate}
\end{observation}

\subsection{Recursive Adversarial Strategy}
Pseudocode~\ref{alg:advstrat} gives the formal description of the recursive adversarial strategy.
The procedure \advstrat{} takes as input the level of recursion $k$ and
the indistinguishable streams $\pi$ and $\rho$ constructed so far.
It also takes two open intervals $(\ell_\pi, r_\pi)$ and $(\ell_\rho, r_\rho)$
of the universe such that so far
there is no item from interval $(\ell_\pi, r_\pi)$ in stream $\pi$
and similarly, $\rho$ contains no item from $(\ell_\rho, r_\rho)$.

The initial call of the strategy for some integer $k$ is $\advstrat$
$(k, \emptyset,$ $\emptyset, (-\infty, \infty),$ $(-\infty, \infty))$,
where $\emptyset$ stands for the empty stream and $-\infty$ and $\infty$ represent
the minimum and maximum items in $U$, respectively.
Note that the assumptions on the input for the initial call are satisfied.
The strategy for $k=1$ is trivial: We just append  $2/\eps$ arbitrary
items from $(\ell_\pi, r_\pi)$ to $\pi$ and any $2/\eps$ items from $(\ell_\rho, r_\rho)$
to $\rho$, in the same order for both streams.
For $k > 1$, we first use \advstrat{} recursively for level $k-1$.
Then, we apply procedure \refineintervals{} on the streams constructed
after the first recursive call, and we get two new intervals on the extreme regions of the largest gap
inside the current intervals.
Finally, we use \advstrat{} recursively for $k-1$ in these new intervals.
Below, we prove that the assumptions on input for these two recursive calls and
for \refineintervals{} are satisfied.

\begin{figure*}[h]
\vspace{10mm}
  \centering
  \begin{subfigure}{\linewidth}
\centering{\includegraphics[width=\linewidth]{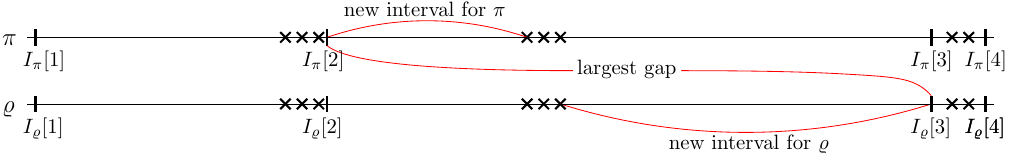}}
\caption{After the first 12 items are sent to $\pi$ and $\rho$}
\label{fig:ex1}
\end{subfigure}
\begin{subfigure}{\linewidth}
\centering{\includegraphics[width=\linewidth]{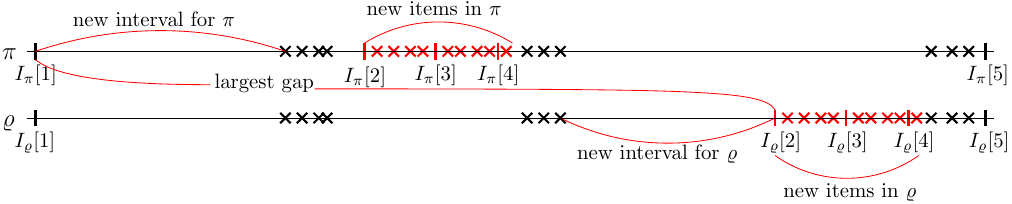}}
\caption{After $24$ items are appended}
\label{fig:ex2}
\end{subfigure}
\begin{subfigure}{\linewidth}
\centering{\includegraphics[width=\linewidth]{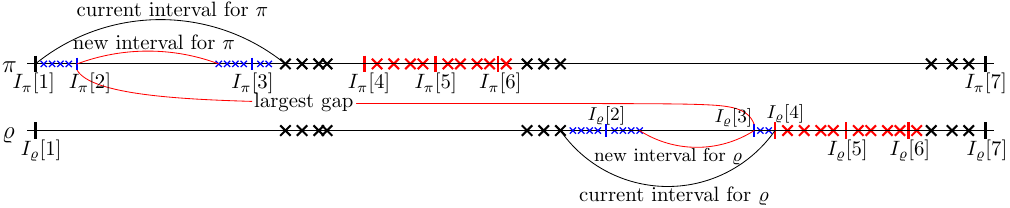}}
\caption{After $36$ items are appended}
\label{fig:ex3}
\end{subfigure}
\begin{subfigure}{\linewidth}
\centering{\includegraphics[width=\linewidth]{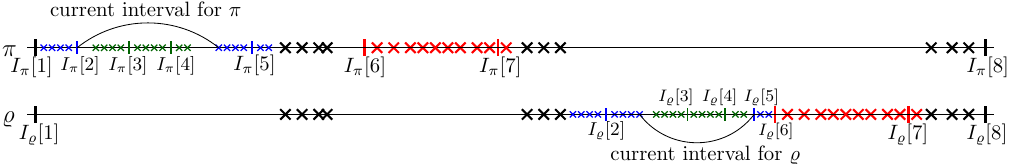}}
\caption{Streams $\pi$ and $\sigma$ with all $N_3 = 48$ items}
\label{fig:ex4}
\end{subfigure}
\caption{An example of the construction of streams $\pi$ and $\sigma$.}
\label{fig:example}
\end{figure*}

\subsection{Example of the Adversarial Strategy}

We now give 
an example of the construction with $k=3$ in Figure~\ref{fig:example}.
The universe is $U = \Re$, which we depict by the real line.
For simplicity, we set $\eps = \frac16$ (although recall that we
require $\eps < \frac{1}{32}$ for our analysis in 
Section~\ref{sec:spaceGapInquality} to hold).

The adversarial construction starts by calling $\advstrat$ $(3, \emptyset,$ $\emptyset,$ $(-\infty, \infty),$ $(-\infty, \infty))$.
The procedure then recursively calls itself twice and in the base case $k=1$,
the two streams $\pi$ and $\rho$ are initialized by $\frac2\eps = 12$
items (we can assume the same items are added to the two streams).
The quantile summary under consideration ($\ds$) chooses to store some of them, but
as $2\eps N_1 = 4$, it cannot forget four consecutive items.

At this point, we are in the execution of $\advstrat$ $(2, \emptyset,$ $\emptyset, (-\infty, \infty),$ $(-\infty, \infty))$,
having finished the recursive call in line~\ref{algLn:stage1}.
Figure~\ref{fig:ex1} shows
the first $12$ items sent to streams $\pi$ and $\rho$, depicted on the real line
for each stream.
A short line segment represents an item that is stored in item array
$I$, while a cross depicts an item not stored by $\ds$.
Note that the largest gap is between the second and the third stored item, i.e.,
$i = 2$ in line~\ref{algLn:gapDef} of \advstrat{}.
This is because $\rank_\pi(I_\pi[2]) = 5$ and  $\rank_\rho(I_\rho[3]) = 9$
(the gap of the same size is also between the first and the second item).
Next, the procedure \refineintervals{} finds the largest gap
and identifies new intervals $(\alpha_\pi, \beta_\pi)$
and $(\alpha_\rho, \beta_\rho)$ for the second recursive call. %

In the execution of  $\advstrat(1, \pi', \rho', (\alpha_\pi,
\beta_\pi), (\alpha_\rho, \beta_\rho))$, there are
$\frac2\eps = 12$ items appended to the streams and the largest gap can be of size at most $2\eps N_2 = 8$.
In Figure~\ref{fig:ex2}, we show the last $12$ items, appended in the second leaf of the recursion tree,  highlighted in red.
Note that fewer of the first $12$ items in the streams are now stored
and that among the $12$ newly appended items, the first, the sixth,
and the eleventh are stored for both streams.
The execution returns to the root node of the recursion tree and
the adversary finds the largest gap together with new intervals.
One of the largest gaps is now between the first and the second stored item
(in this example, all gaps have the same size of $8$).

The execution then goes to the third leaf, where $12$ items are
appended for the third time.
Figure~\ref{fig:ex3} illustrates this, with 
the most recent $12$ items shown smaller and in blue.
In the execution of $\advstrat$ for $k=2$, the largest gap is found
---
note that we look for it only in the current intervals,
and that its size can be at most $2\eps \cdot 3\cdot \frac2\eps = 12$ items.
One of the two largest gaps is between the second and the third item
in the restricted item arrays; these are also the second and the third item
in the whole item arrays
(the other gap of the same size is between the first and the second stored item).
Again, new intervals are identified for the execution of the last leaf of the recursion tree.

Finally, the last $12$ items are appended to the streams, which completes
the construction.
Figure~\ref{fig:ex4} shows the final state, with these last $12$ items
added in green.
The current intervals are with respect to the last leaf of the recursion tree.
\qed

\subsection{Properties of the Adversarial Strategy}
We first give some observations.
Note that the recursion tree of an execution of \advstrat($k$)
indeed has $2^{k-1}$ leaves which each correspond to calling the
strategy for $k=1$, 
and that the items are appended to streams only in the leaves, namely, $\frac{2}{\eps}$ items to each stream in each leaf.
It follows that the number of items appended is $N_k = \oneovereps\cdot 2^k$.
Observe that for a general recursive call of \advstrat, the input streams $\pi$ and $\rho$ 
may already contain some items.
Also, the behavior of comparison-based quantile summary $\ds$ may be different when processing items
appended during the recursive call in line~\ref{algLn:stage1} and when processing items from the call in line~\ref{algLn:stage2}.
The reason is that the computation of $\ds$ is also influenced by items outside the intervals, i.e.,
by items in streams $\pi$ and $\rho$ that are from other branches of the recursion tree.
We remark that items in each of $\pi''$ and $\rho''$ are distinct within the streams
(but the two streams may share some items, which does not affect our analysis).

We now prove that the streams constructed are indistinguishable
and that we do not violate any assumption on input for any recursive call.
We use the following lemma derived from~\cite{hung10_qs_lower_bound}
(which is a simple consequence of the facts that $\ds$ is comparison-based
and the memory states $(I_\pi, G_\pi)$ and $(I_\rho, G_\rho)$ are equivalent).

\begin{lemma}[Implied by Lemma~2 in~\cite{hung10_qs_lower_bound}]\label{lem:indistAppending}
Suppose that streams $\pi$ and $\rho$ are indistinguishable for $\ds$ and
let $I_\pi$ and $I_\rho$ be the corresponding item arrays after processing
$\pi$ and $\rho$, respectively. Let $a, b$ be any two items
such that $\min \{ i | a \le I_{\pi}[i] \} = \min \{ i | b \le I_{\rho}[i] \}$.
Then the streams $\pi a$ and $\rho b$ are indistinguishable.
\end{lemma}

\begin{lemma}\label{lem:streamsIndistinguishable}
Consider an execution of $\advstrat\big(k, \pi, \rho,$ $(\ell_\pi, r_\pi),$ $(\ell_\rho, r_\rho)\big)$ for $k\ge 1$
and let $\pi''$ and $\rho''$ be the returned streams. 
Suppose that streams $\pi$ and $\rho$ and intervals $(\ell_\pi, r_\pi)$ and $(\ell_\rho, r_\rho)$
satisfy the assumptions on the input of $\advstrat$.
Then, for $k > 1$, the assumptions on input for the recursive calls in lines~\ref{algLn:stage1} and~\ref{algLn:stage2}
and for the call of \refineintervals{} in line~\ref{algLn:refineIntervals} are satisfied,
and, for any $k\ge 1$, the streams $\pi''$ and $\rho''$ are indistinguishable.
\end{lemma}

\begin{proof}
The proof is by induction on $k$. 
In the base case $k = 1$,  we use the fact that the $\frac{2}{\eps}$ items from the corresponding intervals are appended in their order
and that $\min \{ i | a \le I_{\pi}[i] \} = \min \{ i | b \le I_{\rho}[i] \}$ for any $a\in (\ell_\pi, r_\pi)$ and $b\in (\ell_\rho, r_\rho)$ by assumption~(iii) on the input of \advstrat.
Thus, applying Lemma~\ref{lem:indistAppending} for each pair of appended items, we get that $\pi''$ and $\rho''$ are indistinguishable.

Now consider $k>1$.
Note that assumptions (i)-(iii) of the first recursive call (in line~\ref{algLn:stage1})
are satisfied by the assumptions of the considered execution.
So, by applying the inductive hypothesis for the first recursive call, streams $\pi'$ and $\rho'$ are indistinguishable.

Next, the assumptions of procedure~\refineintervals{}, called in line~\ref{algLn:refineIntervals},
are satisfied, since streams $\pi'$ and $\rho'$ are indistinguishable,
$\pi$ contains no item from $(\ell_\pi, r_\pi)$,
$\rho$ contains no item from $(\ell_\rho, r_\rho)$,
and the first recursive call in line~\ref{algLn:stage1} generates $N' = \oneovereps\cdot 2^{k-1}$
items from $(\ell_\pi, r_\pi)$ into $\pi'$ and $N'$ items from $(\ell_\rho, r_\rho)$ into $\rho'$.

Then, assumption~(i) of the second recursive call in line~\ref{algLn:stage2}
holds, since $\pi'$ and $\rho'$ are indistinguishable, and
assumptions~(ii) and (iii) are satisfied by applying Observation~\ref{obs:refineIntervalsOutput}.
Finally, we use the inductive hypothesis for the recursive call
in line~\ref{algLn:stage2} and get that streams $\pi''$ and $\rho''$ are indistinguishable.
\end{proof}

Our final observation is that for any $1\le i\le |I_{\pi''}|$,
we have that $\rank_{\pi''}(I_{\pi''}[i]) \le \rank_{\rho''}(I_{\rho''}[i])$.
The proof follows by induction on $k$ (similarly to Lemma~\ref{lem:streamsIndistinguishable}) and 
by the definition of the new intervals in lines~\ref{algLn:gapDef}-\ref{algLn:intervalForRho}
of procedure \refineintervals{},
namely, since the new interval for $\pi$ is in the leftmost region of the largest gap,
while the new interval for $\rho$ is in the rightmost region.

\section{Space-Gap Inequality}\label{sec:spaceGapInquality}

\subsection{Intuition for the Inequality}
In this section, we analyze the space required by data structure $\ds$
when invoked on the two adversarial inputs from the previous section.
Recall that our general goal is to prove that
$\ds$ needs to store $c\cdot \oneovereps$ items from the first half of the whole stream $\pi$ (or, equivalently, from $\rho$)
and $c\cdot \oneovereps\cdot (k-1)$ items from the second half (by using induction on the second half),
where $c > 0$ is a constant.
Note also that if $\ds$ stores $c\cdot \oneovereps\cdot k$ items from
the first half of the stream, the second half of the argument is not even needed.

However, we actually need to prove a similar result for {\em any}
internal node of the recursion tree,  where the bounds as stated above may not hold. 
For instance, $\ds$ may use nearly no space for some part of the stream,
which implies a lot of uncertainty there, but still may be able to provide any $\eps$-approximate
$\phi$-quantile, since the largest gap introduced earlier is very low.
We thus give a space lower bound for an execution of \advstrat{} 
that depends on the largest gap size, denoted $g$, which is introduced in this execution.
Roughly, the space lower bound is $c\cdot (\log g)\cdot  N_k / g$ for a constant $c>0$, so by setting
$g = 2 \eps N_k$ we get the desired result.
For technical reasons, the actual bound stated below
  as~\eqref{eqn:spaceVsGap} is a bit more complicated.
We refer to this bound as the ``space-gap inequality'', and the bulk
of the work in this section is devoted to proving this inequality.

The crucial claim needed in the proof is that, for $k > 1$, the largest gap size $g$
is, in essence, the sum of the largest gap sizes $g'$ and $g''$ introduced in the first and the second recursive call, respectively.
This claim gives two cases: Either the gap $g'$ from the first recursive call
is small (less than approximately half of $g$) and thus $\ds$ uses a lot of space for items from the first recursive call,
or $g' \gtrsim \half g$, so $g'' \lesssim \half g$ and we use induction on the second recursive call, 
together with a straightforward space lower bound for items from the first half of the stream.

Finally, we prove that not only the maximum space usage of the algorithm is large
	but also its \emph{average} space usage during the computation is large --- that is, $\ds$ cannot use a too small space in most steps.

\subsection{Stating the Space-Gap Inequality}

We perform the formal analysis by induction.
We define
\[S(k, \pi, \rho, (\ell_\pi, r_\pi), (\ell_\rho, r_\rho)) := \left|I^{(\ell_\pi, r_\pi)}_{\pi''}\right|,\]
where
\[(\pi'', \rho'') = \advstrat(k, \pi, \rho, (\ell_\pi, r_\pi), (\ell_\rho, r_\rho)).\]
In words, it is the size of the item array restricted to $(\ell_\pi, r_\pi)$
after the execution of $\ds$ on stream $\pi''$ (or, equivalently, with $\rho$ instead of $\pi$).
For simplicity, we write $S_k = S(k, \pi, \rho, (\ell_\pi, r_\pi), (\ell_\rho, r_\rho))$
when $\pi$ and $\rho$ are clear from the context.

Furthermore, let $\Sbar(k, \pi, \rho, (\ell_\pi, r_\pi), (\ell_\rho, r_\rho))$
	be the average size of $I^{(\ell_\pi, r_\pi)}_{\widebar{\pi}}$ over all prefixes $\widebar{\pi}$
	of $\pi''$ that contain an item from $(\ell_\pi, r_\pi)$, for $\pi''$ defined as above, and let $\Sbar_k = \Sbar(k, \pi, \rho, (\ell_\pi, r_\pi), (\ell_\rho, r_\rho))$.
	That is, we take the average size of the item array restricted to $(\ell_\pi, r_\pi)$ during the execution of
	$\advstrat(k, \pi, \rho, (\ell_\pi, r_\pi), (\ell_\rho, r_\rho))$.

We prove a lower bound on $\Sbar_k$ that depends on the largest gap
between the restricted item arrays for $\pi$ and for $\rho$.
We enhance the definition of the gap to take the intervals restriction into account.
\begin{definition}\label{def:gap}
For indistinguishable streams $\sigma$ and $\tau$ and intervals $(\ell_\sigma, r_\sigma)$ and $(\ell_\tau, r_\tau)$,
let $\overline{\sigma}$ and $\overline{\tau}$ be the substreams of $\sigma$ and $\tau$ consisting
only of items from intervals $(\ell_\sigma, r_\sigma)$ and $(\ell_\tau, r_\tau)$, respectively.
Moreover, let $I'_{\sigma} = I^{(\ell_\sigma, r_\sigma)}_{\sigma}$
and $I'_{\tau} = I^{(\ell_\tau, r_\tau)}_{\tau}$ be the restricted item arrays after processing ${\sigma}$ and ${\tau}$, respectively.
We define the \emph{largest gap} between $I'_{\sigma}$ and $I'_{\tau}$ in intervals $(\ell_\sigma, r_\sigma)$ and $(\ell_\tau, r_\tau)$ as
$$\gap\big({\sigma}, {\tau}, (\ell_\sigma, r_\sigma), (\ell_\tau, r_\tau)\big) = \max_{1\le i < |I'_{\tau}|}
	\rank_{\overline{\tau}}(I'_{\tau}[i+1]) - \rank_{\overline{\sigma}}(I'_{\sigma}[i]) \,.$$
\end{definition}

Note that the ranks are with respect to substreams $\overline{\sigma}$ and $\overline{\tau}$, and
that the largest gap is always at least one, supposing that the ranks of stored items
are not smaller for $\tau$ than for $\sigma$.
We again have $\gap\big(\sigma, \tau, (\ell_\sigma, r_\sigma), (\ell_\tau, r_\tau)\big)\ge \gap\big(\sigma, \sigma, (\ell_\sigma, r_\sigma), (\ell_\sigma, r_\sigma)\big)$.
Also, as the restricted item arrays are enclosed by interval boundaries, the following simple bound holds:
\begin{align}
\nonumber
  S_k = S(k, \pi, \rho, (\ell_\pi, r_\pi), (\ell_\rho, r_\rho))
&\ge \frac{N_k}{\gap\big(\pi'', \pi'', (\ell_\pi, r_\pi), (\ell_\pi,
  r_\pi)\big)} \\
&\ge \frac{N_k}{\gap\big(\pi'', \rho'', (\ell_\pi, r_\pi), (\ell_\rho, r_\rho)\big)}\,,
  \label{eqn:trivSpaceBoundMax}
\end{align}
where $(\pi'', \rho'') = \advstrat(k, \pi, \rho, (\ell_\pi, r_\pi), (\ell_\rho, r_\rho))$
and $N_k = \oneovereps\cdot 2^k$.
Furthermore, as the gap size can only increase and moreover, in the second half of the steps, $\pi''$ contains 
at least $N_k / 2$ items, we get that the average size of the restricted item arrays while processing $\pi''$
is 
\begin{align}
	\Sbar_k = \Sbar(k, \pi, \rho, (\ell_\pi, r_\pi), (\ell_\rho, r_\rho)) 
	&\ge \frac{N_k}{2\cdot \gap\big(\pi'', \rho'', (\ell_\pi, r_\pi), (\ell_\rho, r_\rho)\big)}\,.
	\label{eqn:trivSpaceBoundAvg}
\end{align}

The following lemma (proved below) shows a stronger inequality between the space and the largest gap.

\begin{lemma}[Space-gap inequality]\label{lem:spaceVsGap}
Consider an execution of $\advstrat(k, \pi, \rho, (\ell_\pi, r_\pi),$ $(\ell_\rho, r_\rho))$.
Let $\pi''$ and $\rho''$ be the returned streams,
and let $g := \gap\big(\pi'', \rho'',$ $(\ell_\pi, r_\pi),$ $(\ell_\rho, r_\rho)\big)$.
Then, for $\Sbar_k = \Sbar(k, \pi, \rho,$ $(\ell_\pi, r_\pi),$ $(\ell_\rho, r_\rho))$, the following \emph{space-gap inequality} holds with $c = \frac{1}{16} - 2\eps > 0$:
\begin{equation}\label{eqn:spaceVsGap}
\Sbar_k \ge c\cdot (\log_2 g + 1)\cdot \left(\frac{N_k}{g} - \frac{1}{4\eps}\right)\,.
\end{equation}
\end{lemma}

We remark that we do not optimize the constant $c$. 
Note that the right-hand side (RHS) of~\eqref{eqn:spaceVsGap} is non-increasing for integer $g\ge 1$,
as $(\log_2 g + 1) / g$ is decreasing for $g\ge 2$ and equals $1$ for $g\in \{1,2\}$.

First, observe that Theorem~\ref{thm:main} directly follows from
Lemma~\ref{lem:spaceVsGap}, and so our subsequent work will be in
proving this space-gap inequality.
Indeed, consider any integer $k\ge 1$ and let $(\pi, \rho) = \advstrat(k,$ $\emptyset, \emptyset,$ $(-\infty, \infty), (-\infty, \infty))$
be the constructed streams of length $N_k$.
Let $g = \gap\big(\pi, \rho,$ $(-\infty, \infty),$ $(-\infty, \infty)\big) = \gap(\pi, \rho)$.
Since $\pi$ and $\rho$ are indistinguishable by Lemma~\ref{lem:streamsIndistinguishable},
we have $g \le 2\eps N_k$ by Lemma~\ref{lem:indistStreamsGapUB}.
Since the RHS of~\eqref{eqn:spaceVsGap} is decreasing for $g\ge 2$ and $2\eps N_k\ge 2$,
it becomes the smallest for $g = 2\eps N_k$. Thus, by Lemma~\ref{lem:spaceVsGap}, the average memory used over all steps is at least
\begin{align*}\Sbar_k
& \ge c\cdot (\log_2 g + 1)\cdot \left(\frac{N_k}{g} -
\frac{1}{4\eps}\right)
\ge c\cdot (\log_2 2\eps N_k + 1)\cdot \left(\frac{1}{2\eps} - \frac{1}{4\eps}\right)
= \Omega\left( \oneovereps\cdot \log \eps N_k\right)\,.\end{align*}

\subsection{Preliminaries for the Proof of Lemma~\ref{lem:spaceVsGap}}

The proof is by induction on $k$.
First, observe that~\eqref{eqn:spaceVsGap} holds almost immediately if $g\le 2^7$.
As $c < 1/16$, we have $c\cdot (\log_2 g + 1) < 0.5$, and
so by the bound in~\eqref{eqn:trivSpaceBoundAvg},
$\Sbar_k > N_k / (2g) > c \cdot (\log_2 g + 1)\cdot \left(\frac{N_k}{g} - \frac{1}{4\eps}\right)$.
Similarly, if $g \ge 4\eps N_k$, then~\eqref{eqn:spaceVsGap} holds, since 
the RHS of~\eqref{eqn:spaceVsGap} is at most $0$ and $\Sbar_k\ge 0$.\footnote{
	Note, however, that we cannot use Lemma~\ref{lem:indistStreamsGapUB} to show $g\le 2\eps N_k$, since the largest gap
	has size bounded by $2\eps$ times the length of $\pi''$ or $\rho''$, which can be much larger than $N_k$
	(due to items from other branches of the recursion tree).
} We thus assume that $g\in (2^7, 4\eps N_k)$, which immediately implies the base case $k=1$ of the induction,
since $4\eps N_1 = 8 < 2^7$ because $N_1 = \frac{2}{\eps}$.

We now consider $k > 1$.
We refer to streams $\pi, \rho, \pi', \rho', \pi'', \rho''$,
intervals $(\alpha_\pi, \beta_\pi)$ and $(\alpha_\rho, \beta_\rho)$
with the same meaning as in Pseudocode~\ref{alg:advstrat}.
Let $I'_{\pi'} = I^{(\ell_\pi, r_\pi)}_{\pi'}$ and $I'_{\rho'} = I^{(\ell_\rho, r_\rho)}_{\rho'}$
be the restricted item arrays, as in Pseudocode~\ref{alg:refineintervals}.
We make use of the following notation:

\begin{itemize}[nosep]
\item Let $\pi'_{k-1}, \rho'_{k-1}$ be the substreams constructed during the recursive call in line~\ref{algLn:stage1}.
Let $S'_{k-1}$ be the size of $I'_{\pi'}$ (or, equivalently, of $I'_{\rho'}$), and
let $g' = \gap\big(\pi', \rho', (\ell_\pi, r_\pi), (\ell_\rho, r_\rho)\big)$ be the largest gap in the input intervals
after $\ds$ processes one of streams $\pi'$ and $\rho'$.

\item Let $I''_{\pi''} = I^{(\alpha_\pi, \beta_\pi)}_{\pi''}$ and $I''_{\rho''} = I^{(\alpha_\rho, \beta_\rho)}_{\rho''}$
be the item arrays restricted to the new intervals after $\ds$ processes streams $\pi''$ and $\rho''$, respectively.
Let $S''_{k-1}$ be the size of $I''_{\pi''}$, and
let $g'' = \gap\big(\pi'', \rho'',$ $(\alpha_\pi, \beta_\pi),$ $(\alpha_\rho, \beta_\rho)\big)$ be the largest gap in the new intervals.
Let $\pi''_{k-1}$ and $\rho''_{k-1}$ be the substreams constructed during the recursive call in line~\ref{algLn:stage2}.

\item Let $I'_{\pi''} = I^{(\ell_\pi, r_\pi)}_{\pi''}$ and $I'_{\rho''} = I^{(\ell_\rho, r_\rho)}_{\rho''}$
be the item arrays restricted to the input intervals after $\ds$ processes streams $\pi''$ and $\rho''$, respectively.

\item Finally, let $\pi_k$ and $\rho_k$ be the substreams of $\pi''$ and $\rho''$, restricted to
$(\ell_\pi, r_\pi)$ and $(\ell_\rho, r_\rho)$, respectively
(i.e., $\pi_k$ and $\rho_k$ consist of the items appended by the considered execution).
\end{itemize}

We remark that notation $I'$ abbreviates the restriction to intervals $(\ell_\pi, r_\pi)$ and $(\ell_\rho, r_\rho)$ (depending on the stream),
while notation $I''$ implicitly denotes the restriction to the new intervals $(\alpha_\pi, \beta_\pi)$ and $(\alpha_\rho, \beta_\rho)$.
Note that $\pi' = \pi \pi'_{k-1}$, and $\pi'' = \pi' \pi''_{k-1} = \pi \pi_k$, and $\pi_k = \pi'_{k-1}\pi''_{k-1}$,
and similarly for streams $\rho'$, $\rho''$, and $\rho_k$.

We now show a crucial relation between the gaps.

\begin{claim}\label{clm:gapsIneq}
$g \ge g' + g'' - 1$
\end{claim}

\begin{proof}
Define $i$ to be \[i:= \argmax_{1\le i' < |I''_{\pi''}|}\, \rank_{\rho''_{k-1}}(I''_{\rho''}[i'+1]) - \rank_{\pi''_{k-1}}(I''_{\pi''}[i']),\] i.e.,
the position of the largest gap in the arrays $I''_{\pi''}$ and $I''_{\rho''}$.
Let $a := I''_{\pi''}[i]$ and $b := I''_{\rho''}[i+1]$ be the two items
whose rank difference determines the largest gap size.
Note that, while $\ds$ stores $a$ and $b$ in $I''_{\pi''}$ and $I''_{\rho''}$,
these two items do not necessarily need to be stored in $I'_{\pi''}$ and $I'_{\rho''}$, respectively.
This may happen for $a$ only in case $a = \alpha_\pi$ and thus $i = 1$,
and similarly, for $b$ only in case $b = \beta_\rho$ and $i = |I''_{\pi''}| - 1$.
Indeed, for $i > 1$, item $a = I''_{\pi''}[i]$
must be in the whole item array $I_{\pi''}$ and thus also in $I'_{\pi''}$,
and similarly, if $i < |I''_{\pi''}| - 1$, item $b = I''_{\rho''}[i+1]$
must be in $I_{\rho''}$ and thus in $I'_{\rho''}$.
(In the special case $|I''_{\pi''}| = 2$, both $a$ and $b$ may not be in
$I'_{\pi''}$ and in $I'_{\rho''}$, respectively, while if $|I''_{\pi''}| > 2$,
at least one of $a$ or $b$ is actually stored.)

Let $j$ be the largest integer such that $I'_{\pi''}[j] \le a$, and let $a' := I'_{\pi''}[j]$;
by the above observations, $a' = a$ unless $i = 1$ and $a\notin I'_{\pi''}$.
Let $b' := I'_{\rho''}[j+1]$. We now show that $b' \ge b$. 
Indeed, this clearly holds if $b' = b$, so suppose $b'\neq b$. 
This may only happen if $b = \beta_\rho$
is not in $I'_{\rho''}$ and $i = |I''_{\pi''}| - 1$. 
We consider two cases:

\mycase{1} If $a' = I'_{\pi''}[j] \in (\alpha_\pi, \beta_\pi)$,
then $I'_{\rho''}[j]\in (\alpha_\rho, \beta_\rho)$ as $\pi''$ and $\rho''$
are indistinguishable and only the last $N_{k-1}$ items are from these intervals.
Moreover, as $i = |I''_{\pi''}| - 1$, index $j$ is the largest such that $I'_{\rho''}[j]\in (\alpha_\rho, \beta_\rho)$,
thus $b' = I'_{\rho''}[j+1] \ge \beta_\rho = b$.

\mycase{2} Otherwise, $a'\le a = \alpha_\pi$,
which may only happen when $i=1$.
As also $i = |I''_{\pi''}| - 1$, we have $|I''_{\pi''}| = 2$, i.e.,
no items from $(\alpha_\pi, \beta_\pi)$ and from $(\alpha_\rho, \beta_\rho)$
are stored in $I'_{\pi''}$ and in $I'_{\rho''}$, respectively.
Then we have $I'_{\pi''}[j+1] \ge \beta_\pi$, by the definition of $j$. 
Before the second recursive call, it holds that
$\alpha_\pi = I'_{\pi'}[\ell]$ and $\beta_\rho = I'_{\rho'}[\ell+1]$ for some index $\ell$,
i.e., there are $\ell$ items stored in $I'_{\pi'}$ and in $I'_{\rho'}$
which are not larger than $\alpha_\pi$ and $\alpha_\rho$,
respectively. By $a' = I'_{\pi''}[j]\le \alpha_\pi$ and by the definition of $j$,
there are $j\le \ell$ items in $I'_{\pi''}$ no larger than $\alpha_\pi$,
and hence, by indistinguishability of $\pi''$ and $\rho''$,
there are $j$ items in $I'_{\rho''}$ no larger than $\alpha_\rho$.
Since no item in $(\alpha_\rho, \beta_\rho)$ is stored in $I'_{\rho''}$, 
we conclude that $b' = I'_{\rho''}[j+1] \ge \beta_\rho = b$ holds.

\smallskip

To prove the claim, it is sufficient to show
\begin{equation}\label{eqn:gapsIneq_1}
\rank_{\rho_k}(b') - \rank_{\pi_k}(a') \ge g' + g'' - 1\,,
\end{equation}
as the difference on the LHS is taken into account in the definition of $g$.
We have
$$
g'' = \rank_{\rho''_{k-1}}(b) -  \rank_{\pi''_{k-1}}(a)
  \le \rank_{\rho''_{k-1}}(b') - \rank_{\pi''_{k-1}}(a')\,,
$$
since $b'\ge b$ and $a' \le a$.
This rank difference is w.r.t.\ substreams $\pi''_{k-1}$ and $\rho''_{k-1}$,
and we now show that when considering $\pi_k$ and $\rho_k$,
the difference increases by $g'-1$. 
Indeed, as $a' < \beta_\pi$ and $b' > \alpha_\rho$, it holds that 
\[\rank_{\rho''_{k-1}}(b') - \rank_{\pi''_{k-1}}(a') = \rank_{\rho_k}(b') - \rank_{\pi_k}(a') - (g' - 1),\]
using the definitions of $g'$ and of the new intervals in lines~\ref{algLn:gapDef}---\ref{algLn:intervalForRho}
of procedure \refineintervals{} (Pseudocode~\ref{alg:refineintervals}).
Summarizing, we have
$
g'' \le \rank_{\rho''_{k-1}}(b') - \rank_{\pi''_{k-1}}(a')
= \rank_{\rho_k}(b') - \rank_{\pi_k}(a') - (g' - 1)
$,
which shows~\eqref{eqn:gapsIneq_1} by rearrangement.
\end{proof}

\subsection{Completing the Proof of Lemma~\ref{lem:spaceVsGap}}

In the inductive proof of~\eqref{eqn:spaceVsGap} for $k > 1$, we sum the bounds 
from the first and second stages of the adversarial strategy, obtained from the induction hypothesis
and halved as we are now averaging over twice many steps.

The key step is to add a simple space lower bound~\eqref{eqn:trivSpaceBoundAvg} for the $N_{k-1}$ items 
from $\pi'_{k-1}$ for any step during processing $\pi''_{k-1}$;
recall that $\pi'_{k-1}$ and $\pi''_{k-1}$ are the items from the first and second recursive call, respectively, in Pseudocode~\ref{alg:advstrat}. %
Consider a prefix $\overline{\pi}$ of $\pi''$ that contains all items from $\pi'_{k-1}$.
First, observe that the largest gap in $I'_{\overline{\pi}}$, the item array $I_{\overline{\pi}}$ restricted to $(\ell_\pi, r_\pi)$, is at most $g$ --- indeed, if it would be larger than $g$, 
further item insertions cannot make the gap smaller as the rank difference in Definition~\ref{def:gap}
cannot decrease by adding more items satisfying property (iii) in Pseudocode~\ref{alg:advstrat}; 
thus, a gap larger than $g$ in $I'_{\overline{\pi}}$ would contradict the definition of $g = \gap\big(\pi'', \rho'',$ $(\ell_\pi, r_\pi),$ $(\ell_\rho, r_\rho)\big)$.
Furthermore, since there can be two gaps of size at most $g$
around stored items from $\pi''_{k-1}$ (i.e., those in $I''_{\overline{\pi}}$),
the number of items from $\pi'_{k-1}$ stored in $I'_{\overline{\pi}}$ is at least
\begin{equation}\label{eqn:spaceFromStage1}
	\frac{N_{k-1} - 2g}{g} = \frac{N_k - 4g}{2g} \ge \frac{(1 - 16\eps)\cdot N_k}{2g}\,,
\end{equation}
using the assumption that $g\le 4\eps N_k$.

Thus, $\Sbar_k$, the average size of the restricted item arrays while processing $\pi''$, is at least
\begin{align}\label{eqn:spaceVsGap-eq1}
	\frac12\Sbar_k + \frac12\Sbar''_{k-1} + \frac{N_k - 16\eps N_k}{2g}
	\ge 
	&\frac12\cdot c\cdot (\log_2 g' + 1)\cdot \left(\frac{N_{k-1}}{g'} - \frac{1}{4\eps}\right) \nonumber \\
	&+
	\frac12\cdot c\cdot (\log_2 g'' + 1)\cdot \left(\frac{N_{k-1}}{g''} - \frac{1}{4\eps}\right) \nonumber \\
	&+
	\frac{(1 - 16\eps)\cdot N_k}{2g}\,,
\end{align}
and our goal is to show that this is at least $c\cdot (\log_2 g + 1)\cdot \left(\frac{N_k}{g} - \frac{1}{4\eps}\right)$.
Multiplying by $2/c$ and substituting $1/c'$ for $(1 - 16\eps) / c$, we need to prove that for a sufficiently small constant $c' > 0$,
\begin{align}\label{eqn:spaceVsGap-eq2}
	(\log_2 g' + 1)\cdot \left(\frac{N_k}{2g'} - \frac{1}{4\eps}\right)
	+
	(\log_2 g'' + 1)\cdot \left(\frac{N_k}{2g''} - \frac{1}{4\eps}\right)
	+
	\frac{N_k}{c'\cdot g}
	\ge
	2\cdot (\log_2 g + 1)\cdot \left(\frac{N_k}{g} - \frac{1}{4\eps}\right)\,,
\end{align}
Since $g\ge g' + g'' - 1\ge g'$ as $g''\ge 1$ and similarly $g\ge g''$,
adding $(\log_2 g' + 1) / (4\eps) + (\log_2 g'' + 1) / (4\eps)$ to the LHS and  $2\cdot (\log_2 g + 1) / (4\eps)$ to the RHS
makes the inequality stronger as we increase the RHS by at least the increase of the LHS.
After dividing by $N_k$, it suffices to show that
\begin{align}\label{eqn:spaceVsGap-eq3}
	\frac{\log_2 g' + 1}{2g'}	+ \frac{\log_2 g'' + 1}{2g''}   +   \frac{1}{c'\cdot g}
	\ge
	\frac{2\cdot (\log_2 g + 1)}{g}\,,
\end{align}
which we show in the remainder of this section.

First, if $g' \le g/4$, then $(\log_2 g' + 1)/(2g') + 1/(c'\cdot g) \ge 2\cdot (\log_2 g + 1) / g$
for $c' < 1/4$, thus inequality~\eqref{eqn:spaceVsGap-eq3} holds,
and similarly for $g'' \le g/4$.
Otherwise, we have that $g', g'' \ge g/4 \ge 2^5$.

For $x > 0$, let $f(x):=\frac{\log_2 x+1}{2x}$.
Since the function $f(x)$ is decreasing for $x\ge 2$ and 
since $g\ge g' + g'' - 1$ by Claim~\ref{clm:gapsIneq}, we can assume that $g = g' + g'' - 1$, 
and our goal is to prove that
\begin{align}\label{eqn:spaceVsGap-eq4}
	f(g') + f(g - g' + 1)   +   \frac{1}{c'\cdot g} \ge	4 f(g)\,,
\end{align}
The first and second derivatives of $f$ are %
\[
f'(x)
= \frac{1-\ln(2x)}{2x^2\ln2}
\qquad \text{and}\qquad
f''(x)
= \frac{\ln(4x^2)-3}{2x^3\ln2}\,,
\]
where $\ln$ is the natural logarithm.
For $x\ge2$ we have $\ln(4x^2)\ge \ln(16)>3$, so $f''(x)>0$ for all $x\ge2$. 
Also for $x\ge2$, it holds that $\ln(2x)>1$, so $f'(x)<0$.
Hence, $f$ is strictly convex and decreasing on $[2,\infty)$.

Consider $\psi(g') := f(g') + f(g - g' + 1)$ as a function of $g'\in [2, g]$ for a given $g$.
Observe that $\psi(g')$ is still strictly convex, as $f$ is strictly convex, and symmetric around $(g+1)/2$.
Thus, it has a unique minimum at $g' = (g+1) / 2$, and the LHS of~\eqref{eqn:spaceVsGap-eq4} is minimized
when $g' = (g+1) / 2$. Therefore, we only need to show that
\[
	2\cdot f\left(\frac{g+1}{2}\right)  +   \frac{1}{c'\cdot g} \ge	4 f(g)\,,
\]
Using the definition of $f$ and that $\log_2 \frac{g+1}{2} \ge (\log_2 g) - 1$, we get that the LHS is at least
\[
2\cdot \frac{\log_2 \frac{g+1}{2} + 1}{2\cdot \frac{g+1}{2}}  +   \frac{1}{c'\cdot g}
\ge
2\cdot \frac{\log_2 g}{g+1}  +   \frac{1}{c'\cdot g}
\ge 4\cdot \frac{\log_2 g + 1}{2g} = 4f(g)\,,
\]
where the second inequality holds for $c' < 1/5$, using that $g > 2^7$.
Using that $c' = c / (1 - 16\eps)$, it is sufficient to take $c = 1/16 - 2\eps$.
This concludes the proof of Lemma~\ref{lem:spaceVsGap}, and the space lower bound follows.

\section{Corollaries and Conclusions}\label{sec:conclusions}

Our construction closes the asymptotic gap in the space bounds for
deterministic comparison-based quantile summaries and yields
the optimality of the Greenwald and Khanna's quantile summary~\cite{greenwald01_quantile_summaries}.
A drawback of their quantile summary is that it carries out an intricate 
merging of stored tuples, where each tuple consists of a stored item
together with lower and upper bounds on its rank. A simplified (greedy) version, which merges stored tuples
whenever it is possible, was suggested already in~\cite{greenwald01_quantile_summaries},
and according to experiments reported in Luo \etal~\cite{luo16_quantiles_experimental},
it performs better in practice than the intricate algorithm analyzed in~\cite{greenwald01_quantile_summaries}.
It is an interesting open problem whether or not the upper bound
of $\O(\oneovereps\cdot \log \eps N)$ holds for the greedy variant
of the Greenwald and Khanna's algorithm; 
recently, Gribelyuk \etal~\cite{GribelyukSWY24} showed that another greedy variant
of the algorithm does allow for a simple analysis.

\subsection{Finding an Approximate Median}
One of the direct consequences of our result is that finding
an $\eps$-approximate median requires roughly the same space as constructing a quantile summary.
(This can be done similarly for any other $\phi$-quantile as long as $\eps \ll \phi \ll 1-\eps$.)

\begin{theorem}\label{thm:apxMedian}
For any $\eps > 0$ small enough, 
there is no deterministic comparison-based streaming algorithm 
that finds an $\eps$-approximate median in the stream
and runs in space $o(\oneovereps\cdot \log \eps N)$ on any stream of length $N$.
\end{theorem}

\begin{proof}[Proof sketch.]
Let $\ds$ be a deterministic comparison-based data structure
that finds an $\eps$-approximate median.
Consider the streams $\pi$ and $\rho$ constructed by 
the adversarial procedure from Section~\ref{sec:construction}, i.e.,
$(\pi, \rho) = \advstrat$ $(k,$ $\emptyset, \emptyset,$ $(-\infty, \infty),$ $(-\infty, \infty))$.
Let $g = \gap(\pi, \rho)$.
If $g \le 4\eps N_k$, then the analysis in Section~\ref{sec:spaceGapInquality},
with an appropriately adjusted space-gap inequality,
shows that the algorithm uses space $\Omega(\oneovereps\cdot \log \eps N_k)$.
Thus, consider the case $g > 4\eps N_k$, which implies
that there exists $\phi'\in (0,1)$ such that the item array does not store a $2\eps$-approximate $\phi'$-quantile.
If $\phi' < 0.5$, we append $(1-2\phi')\cdot N_k \le N_k$ items to streams $\pi$ and $\rho$ that are smaller than
any item appended so far, and after that the algorithm cannot return 
an $\eps$-approximate median. Otherwise, $\phi' \ge 0.5$ and 
we append $(2\phi'-1)\cdot N_k \le N_k$ items to streams $\pi$ and $\rho$ that are larger than
any item appended so far.
Thus, in this case also an $\eps$-approximate median is not stored.
\end{proof}

\subsection{Estimating Rank}
We now consider data structures for the following \textsc{Estimating Rank} problem,
which is closely related to computing $\eps$-approximate quantiles:
The input arrives as a stream of $N$ items from a totally ordered universe $U$,
and the goal is to design a data structure with small space cost which is able to
provide an $\eps$-approximate rank for any query $q\in U$, i.e.,
the number of items in the stream which are not larger than $q$, up to an additive error of $\pm \eps N$.
Our construction directly implies a space lower bound for comparison-based rank data structures,
which are defined similarly as in Definition~\ref{def:comparisonBased}.\footnote{
We only need to replace item \textit{(iv)} of Definition~\ref{def:comparisonBased} by 
\textit{(iv) Given a query $q\in U$, the computation of $\ds$ is determined solely by the results of
comparisons between $q$ and $I[j]$, for $j= 1,\dots,|I|$,
the number of items stored, and the contents of $G$.}
}

\begin{theorem}\label{thm:estimatingRank}
For any $\eps > 0$ small enough,
there is no deterministic comparison-based data structure for \textsc{Estimating Rank}
which stores $o(\oneovereps\cdot \log \eps N)$ items on any input stream of length $N$.
\end{theorem}

\begin{proof}[Proof sketch.]
Let $\ds$ be a deterministic comparison-based data structure for \textsc{Estimating Rank}.
Consider again the pair of streams
$(\pi, \rho) = \advstrat(k, \emptyset, \emptyset, (-\infty, \infty), (-\infty, \infty))$.
Let $g = \gap(\pi, \rho)$.
The space-gap inequality (Lemma~\ref{lem:spaceVsGap}) holds, using the same proof.
As shown at the beginning of Section~\ref{sec:spaceGapInquality},
if $g\le 2\eps N_k + 2$, then $\ds$ needs to store $\Omega(\oneovereps\cdot \log \eps N_k)$ items
(the $+2$ makes no effective difference in the calculation).
It remains to observe that if $\ds$ provides an $\eps$-approximate rank of any query $q\in U$,
then $g\le 2\eps N_k + 2$.

Indeed, suppose for a contradiction that $g > 2\eps N_k + 2$, which implies that
there is
$1\le i < |I_\pi| = |I_\rho|$ such that $\rank_\rho(I_\rho[i+1]) - \rank_\pi(I_\pi[i]) > 2\eps N_k + 2$.
Let $q_\pi$ be an item which lies in $(I_\pi[i], \next(\pi, I_\pi[i]))$, that is,
just after $I_\pi[i]$ in $U$ ($q_\pi$ exists by our continuity assumption).
Similarly, let $q_\rho$ be an item in $(\prev(\rho, I_\rho[i+1]), I_\rho[i+1])$.
Let $r$ be the rank returned by $\ds$ when run on query $q_\pi$ after processing stream $\pi$.
Observe that $\ds$ returns $r$ also on query $q_\rho$ after processing stream $\rho$,
since $\pi$ and $\rho$ are indistinguishable, $\ds$ is comparison-based, and
the results of comparisons with stored items are the same in both cases.
However, the true ranks satisfy
$\rank_\pi(q_\pi) = \rank_\pi(I_\pi[i]) + 1$ and
$\rank_\rho(q_\rho) = \rank_\rho(I_\rho[i+1]) - 1$, thus
$\rank_\rho(q_\rho) - \rank_\pi(q_\pi) > 2\eps N_k$. It follows that $r$ differs
from $\rank_\pi(q_\pi)$ or from $\rank_\rho(q_\rho)$ by more than $\eps N_k$,
which is a contradiction.
\end{proof}

\subsection{Randomized Algorithms}
\label{sec:randomlb}
We now turn our attention to randomized quantile summaries, which may fail to
provide an $\eps$-approximate $\phi$-quantile, for some $\phi$, with 
probability bounded by a parameter $\delta$. 
Karnin \etal~\cite{karnin16_optimal_rand_quantile_summaries} designed a
randomized com\-pa\-ri\-son-ba\-sed quantile summary with storage cost $\O(\oneovereps\cdot \log \log \frac{1}{\eps\delta})$.
They also proved the matching lower bound,
which however holds only for a certain stream length (depending on $\eps$)
and for $\delta$ exponentially close to $0$. We state it more precisely as follows.

\begin{theorem}[Theorem~6 in~\cite{karnin16_optimal_rand_quantile_summaries}]\label{thm:randomizedPrev}
For any $\eps > 0$ small enough,
there is no randomized com\-pa\-ri\-son-ba\-sed $\eps$-approximate quantile summary
with failure probability less than $\delta = 1 / N!$,
which stores $o(\oneovereps\cdot \log \log \frac{1}{\delta})$ items on any input stream of length
$N = \Theta\left(\oneoverepssquared\cdot \log^2 \oneovereps\right)$.
\end{theorem}

The proof follows from reducing the randomized case
to the deterministic case and using the lower bound of
$\Omega(\oneovereps\cdot \log \oneovereps)$~\cite{hung10_qs_lower_bound},
which holds for streams of length $N = \Theta\left(\oneoverepssquared\cdot \log^2 \oneovereps\right)$.
Suppose for a contradiction that there exists a comparison-based $\eps$-approximate quantile summary
which stores $o(\oneovereps\cdot \log \log \frac{1}{\delta})$ items for $\delta = 1 / N!$.
Note that if failure probability is below $1 / N!$, a randomized comparison-based quantile summary succeeds 
simultaneously for all streams of length $N$ with probability $> 0$ (by the union bound).
More precisely, it succeeds for all permutations of any given set of $N$
distinct items, which is sufficient in the comparison-based model.
Thus, there exists a choice of random bits which provides a correct result for
all streams of length $N$. Hard-coding these bits, we obtain a deterministic algorithm
running in space 
$o(\oneovereps\cdot \log \log \frac{1}{\delta}) = o(\oneovereps\cdot \log \log e^{N \log N}) =  o(\oneovereps\cdot \log N)
= o(\oneovereps\cdot \log \oneovereps)$,
which contradicts the lower bound in~\cite{hung10_qs_lower_bound}.
We remark that the lower bound holds even for finding the median.

Using our lower bound of $\Omega(\oneovereps\cdot \log \eps N)$ for deterministic
quantile summaries, we strengthen the randomized lower bound so that it holds
for any stream length $N$, which in turn gives a higher space bound.
Hence, using the same proof, we obtain:

\begin{theorem}
There is no randomized comparison-based $\eps$-app\-ro\-xi\-ma\-te quantile summary
with failure probability less than $\delta = 1 / N!$,
which stores $o(\oneovereps\cdot \log \log \frac{1}{\delta})$ items on any input stream of length $N$.
\end{theorem}

Note that the lower bound of $\Omega(\oneovereps\cdot \log \log \frac{1}{\delta})$
for randomized quantile summaries trivially holds if $\delta > 0$ is a fixed constant (say, $\delta = 0.01$),
since any quantile summary needs to store $\Omega(\oneovereps)$ items.
It remains an open problem whether or not the lower bound of $\Omega(\oneovereps\cdot \log \log \frac{1}{\delta})$
holds for, e.g., $\delta = 1/\poly(N)$ or $\delta = 1/\polylog(N)$.

\subsection{Relative-Error Quantiles}
Note that the quantiles problem studied in this paper gives a \emph{uniform}
error guarantee of $\eps N$ for any quantile $\phi\in [0,1]$. 
A stronger, relative-error guarantee of $\eps \phi N$ was proposed by 
Cormode \etal~\cite{cormode05_biased_quantiles}, under the name of \emph{biased quantiles}.
Namely, given a query $\phi\in [0,1]$, an $\eps$-approximate relative-error
quantile summary returns a $\phi'$-quantile for some
$\phi' = \left[(1 - \eps)\cdot \phi, (1 + \eps)\cdot \phi\right]$.\footnote{
Strictly speaking, the definition in~\cite{cormode05_biased_quantiles}
is weaker, requiring only to approximate items at ranks $\phi^j\cdot N$
with error at most $\eps\cdot \phi^j \cdot N$ for $j=0, \dots, \lfloor \log_{1/\phi} N \rfloor$
and some parameter $\phi\in (0,1)$ known in advance.}
In other words, when queried for the $k$-th smallest item (where $k = \lfloor \phi N\rfloor$),
the algorithm returns the $k'$-th smallest item for some
$k' \in \left[(1 - \eps)\cdot k, (1 + \eps)\cdot k\right]$.
Note that the relative-error guarantee and the uniform guarantee of $\eps N$ are essentially the 
same for $\phi = \Omega(1)$, up to a constant factor. That is, biased quantiles
provide a substantially stronger guarantee for extreme values of $\phi$ only, e.g., for $\phi = 1/\sqrt{N}$.

Any summary for biased quantiles, even constructed offline,
requires space of $\Omega(\oneovereps\cdot \log \eps N)$, which is the best lower bound
proved so far. This follows by observing that any summary needs to
store the 
$\oneovereps$ smallest items; among the next $\oneovereps$ items, it
should store every other one; 
and more generally, it needs to store $\Omega(\oneovereps)$ items 
among those with ranks between $\frac{2^i}{\eps}$ and $\frac{2^{i+1}}{\eps}$ for any
$i=0,\dots, \log \eps N$; see~\cite{CormodeKLTV23} for a formal proof.
The state-of-the-art upper bounds for the space requirement in the streaming setting
are $\O(\oneovereps\cdot \log^3 \eps N)$, using a deterministic comparison-based
``merge \& prune'' strategy~\cite{Zhang07_biased_quantiles},
and $\O(\oneovereps\cdot \log \eps N\cdot \log |U|)$ for a fixed universe $U$~\cite{Cormode06_biased_quantiles_U},
using a modification of q-digest from~\cite{shrivastava04_bounded_universe_qs}.
The first randomized algorithms were sampling-based,
requiring space of $\O(\frac{1}{\eps^2}\cdot \log \frac{1}{\delta}\cdot \log \eps N)$
in the worst case~\cite{Gupta03_counting_inversions,zhang06_randomized_biased_quantiles}.
Recently, two new randomized algorithms appeared:
ReqSketch~\cite{CormodeKLTV23}, achieving relative error in space  $\O(\frac{1}{\eps}\cdot \sqrt{\log \frac{1}{\delta}}\cdot \log^{1.5} \eps N)$, even in the more general setting of mergeability~\cite{agarwal13_mergeable_summaries},
and a version of ReqSketch with ``elastic compactors''~\cite{GribelyukSWY25} that achieves near-optimal space $\widetilde{\O}(\frac{1}{\eps}\cdot \log \eps N)$, where $\widetilde{\O}$ hides factors poly-logarithmic in $1/\eps, \log 1/\delta,$ and $\log N$;
however, the mergeability properties of the latter sketch are unknown.

We show that our construction from Section~\ref{sec:construction} can be used
to improve the lower bound for $\eps$-approximate biased quantile summaries
by a further $\log \eps N$ factor.
Note that the definition of  comparison-based summaries (Definition~\ref{def:comparisonBased}) 
translates to this setting, as well as Definitions~\ref{def:equivalentStates}
and~\ref{def:indistinguishableStreams} which define equivalent memory states and indistinguishable streams.

\begin{theorem}\label{thm:biasedQuantilesLB}
For any $\eps > 0$ small enough,
there is no deterministic comparison-based $\eps$-approximate relative-error quantile summary 
which stores $o(\oneovereps\cdot \log^2 \eps N)$ items on any input stream of length $N$.
\end{theorem}

\begin{proof}[Proof sketch.]
For an integer $k$, we show that any deterministic comparison-based 
$\eps$-approximate relative-error quantile summary needs to use
$\Omega(\oneovereps\cdot k^2)$ space 
on some stream of length $\oneovereps\cdot 2^k$, so that $k = \Omega(\log \eps N)$.
Our plan is to execute $k-1$ instances of the adversarial strategy, denoted $A_i$ for $i = 1, \dots, k-1$,
with the $i$-th instance consisting of items whose final ranks will be in $(\oneovereps\cdot 2^i, \oneovereps\cdot 2^{i+1}]$ and
starting with $\frac{2}{\eps}$ smallest items which must be stored at all times.
However, we need to interleave these instances and use that the \emph{average} space usage
of the algorithm on $A_i$ is  $\Omega(\oneovereps\cdot i)$ to show the desired space bound\footnote{
	The previous version of this manuscript incorrectly claimed that it suffices to execute the instances
	$A_i$ one after the other, in their order; however, this does not work as the summary for each of these $A_i$'s
	can be compressed to $O(\oneovereps)$ space after $A_i$ ends, thereby achieving total space only $O(\oneovereps\cdot k)$.
}.
Since the average will be over all $\oneovereps\cdot 2^k$ steps for all $i$,
summing over all $i$'s gives the desired lower bound of $\Omega(\oneovereps\cdot k^2)$.

In more detail, the lower bound construction works as follows:
Suppose w.l.o.g.\ that $\frac{2}{\eps}$ is an integer.
The adversary sends $\frac{2}{\eps}$ distinct items that will be smaller than any other item generated later.
Next, it selects $k-1$ disjoint intervals $(\ell_i, r_i)$ for $i = 1, \dots, k-1$ such that
$r_{i-1} < \ell_i < r_i < \ell_{i+1}$.
Then, the adversary executes strategies $A_i := \advstrat(i, \emptyset, \emptyset, (\ell_i, r_i), (\ell_i, r_i))$
for $i = 1, \dots, k-1$ \emph{in parallel} and always adding items to the same streams $\pi, \rho$ as follows:
For $j = 0, \dots, \oneovereps\cdot 2^k - \frac{2}{\eps}$, if the binary representation of $j$ ends with $z$ zeros,
send the next item from strategy $A_{k - 1 - z}$. 
That is, the adversary sends an item from $A_{k-1}$ in every odd step, an item from $A_{k-2}$ in every fourth step, and so on.
The generated streams $\pi, \rho$ are indistinguishable which follows from a similar argument as in Lemma~\ref{lem:streamsIndistinguishable} (i.e., iteratively applying Lemma~\ref{lem:indistAppending}), additionally using that the intervals $(\ell_i, r_i)$ are disjoint.
The length of both $\pi$ and $\rho$ equals $\frac{2}{\eps} + \sum_{i = 1}^{k-1} \oneovereps\cdot 2^i = \oneovereps\cdot 2^k$.

Note that while the generated instances $A_i$ use distinct ranges $(\ell_i, r_i)$,
there is only a single summary processing all of them in parallel.
Nevertheless, the summary executed on $A_i$ for any $i$ with the item array restricted to
interval $(\ell_i, r_i)$ works as a deterministic streaming algorithm with the uniform error for the items from $(\ell_i, r_i)$. Therefore, the space-gap inequality of Lemma~\ref{lem:spaceVsGap} applies; in particular,
the analysis of Section~\ref{sec:spaceGapInquality} is applied to the substream with items in $(\ell_i, r_i)$
for $i = 1, \dots, k - 1$.

Furthermore, the largest gap inside $(\ell_i, r_i)$ can be at most $4\cdot 2^i$ as the total number of items smaller than $r_i$ is $\oneovereps\cdot 2^{i+1}$.
Hence, in $(\ell_i, r_i)$, the algorithm needs to store at least $\Omega\left(\oneovereps\cdot i\right)$ items
on average over all steps of strategy $A_i$. Since the steps of $A_i$ are evenly distributed in streams $\pi, \rho$,
the average over steps restricted to $(\ell_i, r_i)$ is the same as the average over all steps of $\pi, \rho$.
Thus, the average space usage of the algorithm must be
$\Omega\left(\oneovereps\cdot \sum_{i = 1}^{k-1} i\right) = \Omega\left(\oneovereps\cdot k^2\right)$.
\end{proof}

Considering randomized algorithms, the same reduction as for the uniform quantiles
problem in Section~\ref{sec:randomlb} shows
that there is no comparison-based randomized quantile summary running in 
space $o(\oneovereps\cdot \log \eps N\cdot \log \log \frac{1}{\delta})$
that satisfies the relative error with probability more than $1-\delta$
for $\delta = 1/N!$.
Closing the gaps for (deterministic or randomized) relative-error quantiles
remains open.

\para{Acknowledgments.}
The authors are grateful to Hongxun Wu for pointing out an issue in the lower bound for relative-error algorithms (Theorem~\ref{thm:biasedQuantilesLB})
in the previous versions of this paper.
The work was supported by European Research Council grant ERC-2014-CoG 647557.

\bibliographystyle{plain}
\bibliography{references-QS}

\end{document}